\documentclass[runningheads]{llncs}
\usepackage[margin=1.2in]{geometry}
\usepackage{amssymb}
\newcommand{\upOmega}{\mathrm{\Omega}}
\usepackage{auxlib}

\begin{document}

\title{Auction Design for Value Maximizers with Budget and Return-on-spend Constraints\thanks{All authors (ordered alphabetically) have equal contributions and are corresponding authors.}}
\titlerunning{Auction Design for Value Maximizers}
%\iffalse
\author{Pinyan Lu\inst{1,2} \and
Chenyang Xu\inst{3} \and
Ruilong Zhang\inst{4}}
\authorrunning{Pinyan Lu, Chenyang Xu and Ruilong Zhang}
\institute{ITCS, Shanghai University of Finance and Economics, China \and
Huawei TCS Lab, China \and
Shanghai Key Laboratory of Trustworthy Computing, East China Normal University, China \and
Department of Computer Science and Engineering, University at Buffalo, USA \\
\email{lu.pinyan@mail.shufe.edu.cn, cyxu@sei.ecnu.edu.cn, ruilongz@buffalo.edu}}
%\fi

%\author{Anonymous Author(s)}
%\institute{ }
%
\maketitle              % typeset the header of the contribution
\begin{abstract}
The paper designs revenue-maximizing auction mechanisms for agents who aim to maximize their total obtained values rather than the classical quasi-linear utilities. Several models have been proposed to capture the behaviors of such agents in the literature. In the paper, we consider the model where agents are subject to budget and return-on-spend constraints.
The budget constraint of an agent limits the maximum payment she can afford, while the return-on-spend constraint means that the ratio of the total obtained value (return) to the total payment (spend) cannot be lower than the targeted bar set by the agent. 
The problem was first coined by \cite{DBLP:conf/sigecom/BalseiroDMMZ22}. 
In their work, only Bayesian mechanisms were considered. 
We initiate the study of the problem in the worst-case model and compare the revenue of our mechanisms to an offline optimal solution, the most ambitious benchmark. The paper distinguishes two main auction settings based on the accessibility of agents' information: \emph{fully private} and \emph{partially private}. In the fully private setting, an agent's valuation, budget, and target bar are all private. We show that if agents are unit-demand, constant approximation mechanisms can be obtained; while for additive agents, there exists a mechanism that achieves a constant approximation ratio under a large market assumption. The partially private setting is the setting considered in the previous work~\cite{DBLP:conf/sigecom/BalseiroDMMZ22} where only the agents' target bars are private. We show that in this setting, the approximation ratio of the single-item auction can be further improved, and a $\upOmega(1/\sqrt{n})$-approximation mechanism can be derived for additive agents.

%This setting is motivated by the scenario  of automated bidding and 
%[Balseiro et al., EC 2022]. 
%However, in the previous work, only Bayesian mechanisms were considered, where the reported ratio bars are assumed to be drawn from public distributions. 
% Several auction settings are considered in the paper.

% We investigate our model in three auction environments: (1) single item auction, (2) multiple items auction among unit-demand agents, and (3) multiple items auction among additive agents. 
% For the first two settings, we show that constant approximation can be obtained; while for the last setting, we give a mechanism with approximation ratio $\upOmega(1/\sqrt{n})$, where $n$ is the number of agents. 

\keywords{Auction Design  \and Value Maximizers \and Return-on-spend Constraints.}
\end{abstract}

\section{Introduction}
%In an auction, the auctioneer decides the allocation $\mathbf{x}=(x_1,x_2, \cdots, x_n)$ of the items and the payments $\mathbf{p}=(p_1,p_2, \cdots, p_n)$ to different bidders. The bidder $i$’s utility depends on both the set of items he gets $x_i$  and the payment $p_i$ he makes to the auctioneer. The valuation of the items for him is usually denoted by a valuation function $v_i(x_i)$. How to combine the valuation and payment to get his finaly utility is a tricky modeling problem. 

In an auction with $n$ agents and $m$ items, the auctioneer decides the allocation $\vx=\{x_{ij}\}_{i\in [n],j\in [m]}$ of the items and the agents' payments $\vp=\{p_i\}_{i\in [n]}$.
The agent $i$'s obtained value is usually denoted by a valuation function $v_i$ of the allocation;
%Given a valid allocation and payment rule, the obtained value of agent $i$ is usually formulated by a utility function.
while the agent's utility depends on both the obtained value and the payment made to the auctioneer.
Combining the valuation and payment to get the final utility function is a tricky modeling problem. 

In the classic auction theory and the vast majority of literature from the algorithmic game theory community, one uses the quasi-linear utility function $u_i=v_i-p_i$, i.e., the utility is simply the obtained value subtracting the payment. 
This natural definition admits many elegant mathematical properties and thus has been widely investigated in the literature (e.g.~\cite{DBLP:journals/mor/Myerson81,DBLP:journals/geb/NisanR01,DBLP:books/cu/NRTV2007}).
However, as argued in some economical literature~\cite{baisa2017auction,zhou2021multi}, this utility function may fail to capture the agents' behaviors and thus cannot fit reality well in some circumstances. 
In these circumstances, one usually uses a generic function $u=f(v,p)$ (with monotonicity and possibly convexity properties) to model the utility function. 
Such treatment is surely general enough, but usually not explicitly enough to get a clear conclusion. 
In particular, designing non-trivial truthful mechanisms for agents with a generic and inexplicit utility function is difficult.

%in the emerging auto-bidding world.
%In other words, the agent in the auto-bidding world aims to maximize the valuation alone rather than the difference between valuation and payment.

%In classic auction theory and vast majority literatures from algorithmic game theory community, one uses the quasilinear utility function $u_i=v_i-p_i$, namely one’s utility is simply his value subtracts the payment. This is the most natural one and also has many elegant mathematical properties. However, as argued in some economical literatures~\cite{}, quasilinear utility may not fit the reality well in some circumstances. In these literatures, one usually uses a generical function $u=f(v,p)$ (with monotonicity and possibly convexity properties) to model the utility function. Such treatment is of course general enough, but usually not explicitly enough to get clear conclusion. In particular, it is very difficult to design non-trivial truthful mechanisms for agents with generic and inexplicit utility function.
%\chenyang{to do here.}
Is there some other explicit utility function (beyond the quasi-linear one) that appears in some real applications? One simple and well-studied model is agents with budget constraints (e.g.~\cite{DBLP:conf/sigecom/BalseiroKMM20,DBLP:journals/jet/CheG00,DBLP:conf/sigecom/DevanurW17,DBLP:journals/jet/PaiV14}). In this setting, besides the valuation function, an agent $i$ also has a budget constraint $B_i$ for the maximum payment he can make. In the formal language, the utility is
$$
u_i:= \left\{
\begin{aligned}
&v_i -p_i \;\;\;\; & \text{if $p_i\leq B_i$ ,}   \\
&-\infty\;\;\;\; & \text{otherwise.} \\
\end{aligned}
\right.
$$

In mechanism design, the valuation function $v_i(\cdot)$ is considered as the private information for agent $i$. 
Thus the auctioneer needs to design a truthful mechanism to incentivize the agents to report their true information.
For these models beyond the simplest quasi-linear utility, other parameters might be involved in the agents’ utility functions besides the valuation function, such as budget $B$ in the above example. 
For the mechanism design problem faced by the auctioneer, one can naturally ask the question of whether these additional parameters are public information or private. 
Both cases can be studied, and usually, the private information setting is more realistic and, at the same time, much more challenging. 
This is the case for the budget constraint agents. 
Both public budget and private budget models are studied in the literature (e.g.~\cite{DBLP:conf/icalp/DobzinskiL14,DBLP:conf/sigecom/BorgsCIMS05,DBLP:conf/sigecom/ChawlaMM11,DBLP:conf/focs/DobzinskiLN08,DBLP:conf/sigecom/LuX15}).

\paragraph{Value Maximizer.}
The above budget constraint agent is only slightly beyond the quasi-linear model since it is still a quasi-linear function as long as the payment is within the budget. However, it is not uncommon that their objective  is to maximize the valuation alone rather than the difference between valuation and payment for budget-constrained agents. 
This is because in many scenarios the objective/KPI for the department/agent/person who really decides the bidding strategy is indeed the final value obtained. 
On the other hand, they cannot collect the remaining unspent money by themselves anyway, and as a result, they do not care about the payment that much as long as it is within the budget given to them. 
For example, in a company or government's procurement process,  the agent may be only concerned with whether the procurement budget can generate the maximum possible value. We notice that with the development of modern auto-bidding systems, value maximization is becoming the prevalent behavior model for the bidders~\cite{DBLP:conf/wine/AggarwalBM19,DBLP:conf/innovations/Babaioff0HIL21,DBLP:conf/nips/BalseiroDMMZ21,DBLP:conf/www/DengMMZ21}. This motivates the study of value maximizer agents, another interesting explicit non-quasi-linear utility model. In many such applications, there is another return-on-spend (RoS) constraint
$\tau_i$ for each agent $i$ which represents the targeted minimum ratio between the obtained value and the payment and is referred to as the \emph{target ratio} in the following. Formally, the utility function is 
\begin{equation}\label{eq:utility}
    u_i:= \left\{
\begin{aligned}
&\;\;\;v_i \; \;\;\;\;\;\text{if $p_i\leq B_i$ and $p_i\tau_i \leq v_i$ ,} \\
&-\infty\;\;\;\;\text{otherwise.} \\
\end{aligned}
\right.
\end{equation}
% \iffalse
% As you can see,
% it is still a function of $v$ and $p$ but with two additional parameters: the budget $B$ and targeted RoS ratio $\tau$.  This utility function is identical to that of~\cite{DBLP:conf/sigecom/BalseiroDMMZ22}. 
% In their paper, they focused on one particular setting where both valuations and budgets are public information, with the target (RoS) ratios being the only single-dimensional private information.
% However, in real business, we think that the setting where all three types of information are private is the most relevant one. Of course, it is also the most challenging one. Therefore, we consider all these different settings with a focus on the all private information setting.\chenyang{Mark}
% \fi

As one can see, the value maximizer's utility function is still a function of $v$ and $p$ but with two additional parameters $B$ and $\tau$, which result from two constraints.
Note that the above utility function is identical to that of~\cite{DBLP:conf/sigecom/BalseiroDMMZ22}.
Their paper focused on one particular setting where both value and budget are public information, with RoS parameter $\tau$ being the only single-dimensional private information.
Considering $\tau$ as the only private information helps design better auctions, but it may fail to capture more wide applications.
On top of capturing more practical applications, we 
 consider the setting where all these pieces of information are private, which we call the fully private setting. 
This makes designing an efficient auction for the problem challenging. 
With the focus on the fully private setting, we also consider some partially private settings, for which we can design better mechanisms. 

There are other definitions of value maximizer in the literature, most of which can be viewed as a special case of the above model~\cite{DBLP:journals/corr/abs-2211-16251,DBLP:conf/www/CavalloKSW17}. 
For example, there might be no budget constraint ($B=\infty$) or no RoS constraint. Another example is to combine $\frac{v_i}{\tau_i}$ as a single value (function). 
A mechanism for the fully private setting in our model is automatically a mechanism with the same guarantee in all these other models. 
That is another reason why the fully private setting is the most general one.  

\paragraph{Revenue maximization and benchmarks}
This paper considers the revenue maximization objective for the auctioneer when designing truthful\footnote{A mechanism is \emph{truthful} if for any agent $i$, reporting the true private information always maximizes the utility regardless of other agents’ reported profiles, and the utility of any truthtelling agent is non-negative.} mechanisms for value maximizers.
For the revenue maximization objective, there are usually two benchmarks, called ``\emph{first-best}" and ``\emph{second-best}".
The first-best benchmark refers to the optimal objective we can get if we know all the information. In our setting, it is $\max_{\vx} \sum_i \min \left\{B_i, \frac{v_i(\vx)}{\tau_i}\right\}$.
% \[\max_{\vx} \sum_i \min \left\{B_i, \frac{v_i(\vx)}{\tau_i}\right\} .\]
For the traditional quasi-linear utility function, the first-best benchmark is simply the maximum social welfare one can generate $\max_{\vx} \sum_i v_i(\vx)$. It is proved that such a benchmark is not achievable or even not approximated by a constant ratio in the traditional setting.
Thus the research there is mainly focused on the second-best benchmark. The second-best benchmark refers to the setting where the auctioneer additionally knows the distribution of each agent's private information and designs a mechanism to get the maximum expected revenue with respect to the known distribution. The benchmark in \cite{DBLP:conf/sigecom/BalseiroDMMZ22} is also this second-best benchmark and they provide optimal mechanism when the number of agents is at most two.

It is clear that the first-best benchmark is more ambitious and more robust since it is prior free. 
They focus on the second-best in the traditional setting because the first-best is not even approximable. 
In our new value maximizer agents setting, we believe it is more important to investigate if we can achieve the first-best approximately. 
Thus, we focus on the first-best benchmark in this paper. This is significantly different from that of \cite{DBLP:conf/sigecom/BalseiroDMMZ22}.

\subsection{Our Results}
\label{sec:pre}

\paragraph{Problem Formulation.} 
The formal description of the auction model considered in the paper follows. 
One auctioneer wants to distribute $m$ heterogeneous items among $n$ agents. 
Each agent $i\in [n]$ has a value $v_{ij} $ per unit for each item $j\in [m]$ and a budget $B_i $, representing the maximum amount of money agent $i$ can pay. 
The agent also has a RoS constraint $\tau_i$, representing the minimum ratio of the received value (return) to the total payment (spend) that she can tolerate. 
%Per the literature convention, we refer to this constraint as the return-on-spend (RoS) constraint. 
As mentioned above, several settings of the type (public or private) of $\left( \vB,\vecv,\vtau \right)$ are considered in the paper. 
Agents are value maximizers subject to their budget constraints and RoS constraints (see \cref{eq:utility} for the formula).
% \begin{equation}
% u_i:= \left\{
% 	\begin{aligned}
% 	v_i \;\;\;\;& \text{if $p_i\leq B_i$ and $p_i\tau_i \leq v_i$ }&  \\
% 	-\infty\;\;\;\; & \text{otherwise.}& \\
% 	\end{aligned}
% 	\right.
% \end{equation}
The auctioneer aims to design a truthful mechanism that maximizes the total payment.
%\ruilong{may need to mention the definition of truthfulness and approximation ratio}

%\subsection*{}

We investigate our model in a few important auction environments. 
We studied both indivisible and divisible items, both the single-item and the multiple-item auctions. 
When there are multiple items, we consider the two most important valuations: unit demand and additive.  
Unit demand models are the setting where the items are exclusive to each other. 
Additive models are the setting where an agent can get multiple items and their values add up. We leave the more generic valuation function, such as submodular or sub-additive, to future study. 
%See~\cref{tab:results} for an overview. 

\begin{table}[tb]\label{tab:results}
\centering
\caption{An overview of our results. We use divisibility ``Y" or ``N" to represent whether items are divisible or not, and use notation $\Leftarrow$ to express the result can be directly implied by the left one. The superscript $\dag$ in the table means that the approximation ratio is obtained under an assumption.}
\begin{tabular}{cccc}
\hline
                             & Divisibility & Fully Private        & Partially Private    \\ \hline
\multirow{2}{*}{Single Item} & Y          &     $\frac{1}{52}$ (\cref{thm:single_all_pvt})                  &   $ \frac{1}{2}-\delta$ (\cref{thm:single})                    \\ \cline{2-4} 
                             & N           &       $\opt$ (\cref{thm:single_indivisible})              &            $\Longleftarrow$          \\ \hline
\multirow{2}{*}{\begin{tabular}[c]{@{}c@{}}Multiple Items with \\ Unit Demand Agents\end{tabular}} & Y & $\upOmega(1)$ (\cref{thm:unit_demand}) & $\Longleftarrow$ \\ \cline{2-4} 
                             & N            &  $\frac{1}{2}$ (\cref{thm:multi_unit_indivisible}) & $\Longleftarrow$ \\ \hline
\multirow{2}{*}{\begin{tabular}[c]{@{}c@{}}Multiple Items with \\ Additive Agents\end{tabular}}    & Y & $\upOmega(1)^\dag$ (\cref{thm:large_market}) & $\upOmega(\frac{1}{\sqrt{n}})$ (\cref{thm:additive})     \\ \cline{2-4} 
                             & N           &        Open              & Open                  \\ \hline
\end{tabular}
\end{table}

In the fully private information setting, we obtain  constant approximation truthful mechanisms for both the single-item auction and the multiple items auction among unit demand agents. 
This is quite surprising given the fact that such a constant approximation to the first-best benchmark is proved to be impossible for the classic quasi-linear utility agents even in the single-item setting. 
The intuitive reason is that the agent is less sensitive to the payment in the value maximizer setting than in the quasi-linear utility setting and thus the auctioneer has a chance to extract more revenue from them.  
But this does not imply that designing a good truthful mechanism is easy. 
Quite the opposite, we need to bring in some new design and analysis ideas since the truthfulness here significantly differs from the traditional one as agents’ utility functions are different.
For the additive valuation, we provide constant approximation only under an additional large market assumption. This is obtained by observing an interesting and surprising relationship between our model and the model of ``liquid welfare for budget-constrained agents". 

We also consider the partially private information setting. 
For the public budget (but private value and target ratio), we obtained an improved constant approximation truthful mechanism for the single-item environment. 
The improved mechanism for the single-item setting has a much better approximation since we cleverly use the public budget information in the mechanism. 
For the additive valuation without the large market assumption, we also investigate it 
in the private target ratio (but public budget and valuation) setting, which is the setting used in \cite{DBLP:conf/sigecom/BalseiroDMMZ22}. we obtained an  $\upOmega(\frac{1}{\sqrt{n}})$ approximation truthful mechanism.  In the additive setting, an agent may get multiple items, and thus the payment she saved from one item can be used for other items, which is an impossible case in the unit demand setting.
Due to this reason, agents may become somewhat more sensitive to payment which leads to an $\upOmega(\frac{1}{\sqrt{n}})$ approximation.
% However, we only get a much worse approximation for additive agents even in this public valuation and budget environment. The intuitive reason why the additive setting is more difficult than the unit demand setting is that an agent can get multiple items and thus the payment she saved from one item can be used for other items. Due to this reason, it becomes somewhat more sensitive to payment.

% \subsection*{Conclusion and open directions}

% We investigate the emerging value maximizer in recent literatures but also departure from their modeling significantly. We believe that the model and benchmark proposed in this paper is on one hand more realistic and on the other hand friendlier to the AGT community. We get a few non-trivial positive results which indicates that this model and benchmark is indeed tractable. There are also many more open questions left.  For additive valuation, it is open if we can get a constant approximation. For single item and unit-demand, it is interesting to design mechanism with better approximation since our current ratio is fairly large. No lower bound is obtained in this model. Any non-trivial lower bound is interesting. For single item setting, we get a much better approximation ratio when valuation and budget are public than the all private setting. However, this is not a separation yet since we do not have any lower bound. Any separation result for different information model in terms which are public and which are private is interesting.

\subsection{Related Works}

The most relevant work is~\cite{DBLP:conf/sigecom/BalseiroDMMZ22}, in which they also aim to design a revenue-maximizing Bayesian mechanism for value maximizers with a generic valuation and utility function under budget and RoS constraints. As mentioned above, they focus on the setting where each agent's only private information is the target ratio, which is referred to as the partially private setting in our paper. They show that under the second-best benchmark, an optimal mechanism can be obtained for the two-agent case.
%They obtained an optimal mechanism when the benchmark is the Bayesian setting (the second-best setting).
% The main difference is that their benchmark is the Bayesian setting (the second-best setting) while our benchmark is first
%\pinyan{to do}

Another closely related line of work is ``liquid welfare for budget constraint agents"~\cite{DBLP:conf/icalp/DobzinskiL14,DBLP:conf/sagt/LuX17,DBLP:conf/sigecom/LuX15,DBLP:journals/mor/CaragiannisV21}. 
We observe an interesting and surprising relationship between these two models since the liquid welfare benchmark is almost identical to the first-best benchmark in our setting. Therefore, some algorithmic ideas there can be adapted here. 
However, there are two significant differences: (\rom{1}) the objective for the auctioneer is (liquid) welfare rather than revenue. 
This difference mainly affects the approximation;
(\rom{2}) the bidders are quasilinear utility (within the budget constraint) rather than value maximizers.  
This difference mainly affects truthfulness. 
Observing this relation and difference, some auction design ideas from their literature inspire part of our methods.
Furthermore, building deeper connections or ideal black-box reductions between these two models would be an interesting future direction.

The model of budget feasible mechanism~\cite{DBLP:conf/focs/Singer10,DBLP:journals/teco/JalalyT21,DBLP:conf/stoc/BeiCGL12,DBLP:conf/soda/ChenGL11,DBLP:journals/algorithmica/LeonardiMSZ21} also models the agent as a value maximizer rather than a quasi-linear utility maximizer as long as the payment is within the budget. 
The difference is that the value maximizer agent is the auctioneer rather than the bidders.

\subsection{Paper Organization}

In the main body, we focus on the fully private setting, where all the budgets, valuations, and target ratios are private. 
We first consider the single-item auction in \cref{subsec:single_all_pvt} and then extend the algorithmic ideas to the multiple items auction for unit demand agents in \cref{sec:unitdemand}. 
Both of the two environments can be constant-approximated. 
Finally, we turn to the multiple items auction for additive agents, the most challenging environment, and show a constant approximation under an assumption on the budgets in~\cref{sec:main_large}.

For the partially private setting, due to space limit, we defer all the results to the appendix. In~\cref{subsec:single_add}, we show that a better constant approximation for the single item environment can be obtained when the budgets become public. Then we leverage this new mechanism to give an  $\upOmega(\frac{1}{\sqrt{n}})$ approximation for multiple items auction among additive agents in~\cref{subsec:add}.

%\input{30-singleitem.tex}

%\section{Fully Private Setting}\label{sec:full_private}

%\chenyang{mark}
%Consider fully setting; then single item; then unit demand; then additive

\section{Warm-up: Single Item Auction}\label{subsec:single_all_pvt}

\begin{algorithm}[tb]
	\caption{\; Single Item Auction }
	\label{alg:single_all_pvt} 	
	\begin{algorithmic}[1] %[1] enables line numbers
		\Require The reported budgets $\{B_i\}_{i\in [n]}$, the reported value profile $\{v_{i}\}_{i\in[n]}$, and the reported target ratios $\{\tau_i\}_{i\in [n]}$.
		\Ensure An allocation and payments.
            \State Initially, set allocation $x_i\leftarrow 0$ $\forall i\in [n]$ and payment $p_i\leftarrow 0$ $\forall i\in [n]$.
		\Begin{With probability of $\frac{9}{13}$,}\Comment{{\color{gray}Indivisibly Selling Procedure}}
            \State Find the agent $k$ with the maximum $\min\left\{ B_k,\frac{v_k}{\tau_k} \right\}$, and break the ties in a fixed manner.
            \State Set $x_k\leftarrow 1, p_k\leftarrow \min\left\{ B_k,\frac{v_k}{\tau_k} \right\}$.
            \End
		\Begin{With probability of $\frac{4}{13}$,} \Comment{{\color{gray}Random Sampling Procedure}}
		\State Randomly divide all the agents with equal probability into set $S$ and $R$. 
		\State Compute the offline optimal solution $\left(\vz^S = \{z_{i}^S\}, \vp(\vz^S)=\{p_i(\vz^S)\}\right)$ of selling the item to the agent subset $S$.
		\State Set the item's reserve price $r \leftarrow \frac{1}{4}\sum_{i\in S}p_i(\vz^S)$.
		\State Let agents in $R$ come in an arbitrarily fixed order. When each agent $i$ comes, use $\alpha$ to denote the remaining fraction of the item, and set $x_{i}\leftarrow \min\left\{\frac{B_i}{r}, \alpha\right\}, p_i\leftarrow x_i\cdot r$ if $r\leq \frac{v_i}{\tau_i}$. 
		\End
		\State \Return Allocation $\{x_{i}\}_{i\in[n]}$ and payments $\{p_i\}_{i\in[n]}$.
	\end{algorithmic}
\end{algorithm}

Let us warm up by considering the environment where the auctioneer has only one item to sell. Our first observation is that if the item is indivisible, we can achieve a truthful optimal solution by directly assigning the item to agent $k$ with the maximum $\min\left\{B_k,\frac{v_k}{\tau_k}\right\}$ and charging her that value. Basically, the first price auction with respect to $ \min\left\{ B_i,\frac{v_i}{\tau_i} \right\}$.
The optimality is obvious. 
For truthfulness, since $ \min\left\{ B_i,\frac{v_i}{\tau_i} \right\}$ is the maximum willingness-to-pay of each agent $i$, if someone other than $k$ misreports the profile and gets assigned the item, one of the two constraints must be violated. 
On the other hand, misreporting a lower profile can only lead to a lower possibility of winning but without any benefit. 

\begin{theorem}\label{thm:single_indivisible}
    There exists a truthful optimal mechanism for the single indivisible item auction.
\end{theorem}

%no one has the incentive to lie due to the two hard constraints.
%If someone other than agent $k$ misreports a profile and obtains the item, one of the constraints must be violated. 
The above theorem gives some intuition for the divisible item environment. If the indivisible optimum is at least a constant fraction $c$ of the divisible optimum, selling the item indivisibly can give a constant approximation. We refer to this idea as \emph{indivisibly selling} in the following.
In contrast, for the case that the indivisible optimum is smaller than a constant fraction $c$ of the divisible optimum (denoted by $\opt$ in the following), we have $\min \left\{ B_i,\frac{v_i}{\tau_i} \right\} \leq c \cdot \opt$ for any agent $i$. 
This property implies that the \emph{random sampling} technique can be applied here. More specifically, we randomly divide the agents into two groups, gather information from one group, and then use the information to guide the item's selling price for the agents in the other group.
Since in an optimal solution, each agent does not contribute much to the objective, a constant approximation can be proved by some concentration inequalities
%Then to obtain a constant approximation in either case, we use a random combination of the two strategies.
based on the above two strategies, we give our mechanism in~\cref{alg:single_all_pvt}.  %It is a random combination of the two strategies.

%claim the following theorem.

\begin{theorem}\label{thm:single_all_pvt}
    \cref{alg:single_all_pvt} is feasible, truthful, and achieves an expected approximation ratio of $\frac{1}{52}$.
\end{theorem}

\begin{proof}
    The feasibility is obvious. Firstly, since $x_i\leq \alpha$ when each agent $i$ comes, $\sum_{i\in [n]}x_i\leq 1$. Secondly, due to $x_i\leq \frac{B_i}{r}$ for each agent $i$, $p_i=x_i \cdot r \leq B_i$. Thirdly, for each agent $i$, we have $x_iv_i\geq p_i\tau_i$ because an agent buys some fractions of the item and gets charged only if $r\leq \frac{v_i}{\tau_i}$.

    Then we show that regardless of which procedure is executed, the mechanism is truthful. The truthfulness of the first procedure is proved by~\cref{thm:single_indivisible} directly. For the second procedure, we show that agents in neither $S$ nor $R$ have the incentive to lie. For an agent in $S$, she will not be assigned anything, and therefore, misreporting her information cannot improve her utility; while for the agents in $R$, they are also truthtelling because their reported information determines neither the arrival order nor the reserve price, and misreporting a higher $\frac{v_i}{\tau_i}$ (resp. a larger $B_i$) to buy more fractions of the item must violate the RoS (resp. budget) constraint of agent $i$.

    Finally, we analyze the approximation ratio. Let $(\vz^*,\vp^*)$ be an optimal solution. Use $\opt$ and $\alg$ to denote the optimal payment and our total payment, respectively. Without loss of generality, we can assume that $p_i^*=x_i^* \cdot \frac{v_i}{\tau_i} \leq B_i$.
    Clearly, if there exists an agent $l$ with $p_l^*\geq  \frac{1}{36}\opt$, we can easily bound the expected total payment by the first procedure:
    \[\E(\alg) \geq \frac{9}{13} \cdot \min\left\{ B_i,\frac{v_i}{\tau_i} \right\} \geq \frac{9}{13}\cdot \min\left\{ B_l,\frac{v_l}{\tau_l} \right\} \geq \frac{1}{52}\opt. \]

    Otherwise, we have $p_i^* < \frac{1}{36}\opt$ $\forall i\in [n]$. Then according to the concentration lemma proved in~\cite[Lemma 2]{DBLP:journals/mor/ChenGL14}, we can establish the relationship between $\sum_{i\in S}p_i^*$ and $\opt$ in the second procedure:
    \begin{equation}\label{eq:single_all_pvt}
        \Pr\left[\frac{1}{3}\opt\leq \sum_{i\in S}p_i^* \leq \frac{2}{3}\opt\right]\geq \frac{3}{4}.
    \end{equation}
    Namely, with probability of at least $3/4$, both $\sum_{i\in S}p_i^*$ and $\sum_{i\in R}p_i^*$ are in $[\frac{1}{3}\opt,\frac{2}{3}\opt]$. 
    
    Let us focus on the second procedure and consider a subset $S$ such that $\sum_{i\in S}p_i^*\in [\frac{1}{3}\opt,\frac{2}{3}\opt]$. We distinguish two cases based on the final remaining fraction of the item.    
    If the item is sold out, our payment is at least $\frac{1}{4}\sum_{i\in S}p_i(\vz^S)$. Since $(\vx^S,\vp(\vz^S))$ is the optimal solution of distributing the item among the agents in $S$, we have
    \[ \alg \geq \frac{1}{4}\sum_{i\in S}p_i(\vz^S)\geq \frac{1}{4}\sum_{i\in S} p_i^* \geq \frac{1}{12}\opt.   \]

    If the procedure does not sell out the item, for any agent $i\in R$ who does not exhaust the budget, $\frac{v_i}{\tau_i} < r = \frac{1}{4}\sum_{i\in S}p_i(\vz^S)$. Using $T\subseteq R$ to denote such agents, we have
    \begin{equation*}
        \begin{aligned}
             \frac{1}{3}\opt \leq \sum_{i\in R}p_i^* & \leq \sum_{i\in R\setminus T} B_i + \sum_{i\in T} p^*_i \leq \alg + \sum_{i\in T} \frac{v_i}{\tau_i} x_i^* \\
             & \leq \alg +  \frac{1}{4}\sum_{i\in S}p_i(\vz^S)\sum_{i\in T} x_i^*  \leq \alg + \frac{1}{4}\sum_{i\in S}p_i(\vz^S) \\
             & \leq \alg + \frac{1}{4}\opt .
        \end{aligned}
    \end{equation*}
    We have $\alg\ge \frac{1}{12}\opt$ from the above inequality.
    
    Thus, in either case, $\alg$ is at least $\frac{1}{12}\opt$ under such a subset $S$. Then according to~\cref{eq:single_all_pvt}, we can complete the proof:
    \[\E(\alg) \geq \frac{4}{13}\cdot \frac{3}{4} \cdot \frac{1}{12}\opt = \frac{1}{52}\opt.  \]
    %completing the proof.    

\end{proof}

\section{Multiple Items Auction for Unit Demand Agents}\label{sec:unitdemand}

This section considers the environment where the auctioneer sells multiple items to unit-demand agents, a set of agents who each desires to buy at most one item. 
We extend the results in the last section and show that a constant approximation can still be obtained. 
Similar to the study of the single-item auction, \cref{subsec:indiv} starts from the indivisible goods environment and shows a $\frac{1}{2}$-approximation.

For the divisible goods environment, our mechanism is also a random combination of the ``indivisibly selling" procedure and the ``random sampling" procedure. However, the mechanism and its analysis are much more complicate than that for single item environment and this section is also the most technical part of this paper. 
We describe the indivisibly selling procedure in~\cref{alg:multi_unit_indivisible}. For the random sampling procedure, the multiple-item setting needs a variant of greedy matching (\cref{alg:greedy}) to compute the reserved prices of each item and \cref{subsec:random} has a discussion about this algorithm. Finally, \cref{subsec:final} analyzes the combined mechanism (\cref{alg:unit_demand}). In order to analyse the approximation ratio of \cref{alg:unit_demand}, we introduce \cref{alg:aux_unit_demand}, a non-truthful mechanism and purely in analysis, to bridge \cref{alg:unit_demand} and \cref{alg:greedy}. 

%Similar to the study of the single-item auction, we also start from the indivisible goods environment. We show that if the items are indivisible, a truthful constant approximation can be obtained by a simple greedy matching method (\cref{subsec:indiv}). Then we consider the other hand and introduce the random sampling strategy for the multiple items auction in~\cref{subsec:random}. Finally, we build on the indivisibly selling idea and the random sampling idea to give the final mechanism in~\cref{subsec:final}. Again, it is a random combination of the two strategies, and we show that by setting the probability of using each strategy appropriately, the mechanism can obtain a constant approximation.

%We first warm up by considering the case where the auctioneer has only one item to sell in~\cref{subsec:single_all_pvt}. Then \cref{subsec:indiv} and \cref{subsec:random} state the aforementioned two strategies, respectively. Finally, we present the final mechanism in~\cref{subsec:final}.

%Hence, randomly selecting half of the agents and only allowing them to buy items

%randomly dividing the agents into two groups and allowing only one group to buy items will not decrease the optimal payment too much. Hence,  can gather information from one group and use the information to guide the items' selling prices for the agents in the other group. 

%In the following,  %when each item is indivisible.    

%Clearly, if all the target ratios are public, we can easily obtain the optimal solution by solving a linear program.

\subsection{Indivisibly Selling}\label{subsec:indiv}

%This subsection provides more details about the indivisibly selling idea. 
%This subsection starts to consider multiple items auction.
We first prove the claimed truthful constant approximation in the scenario of selling indivisible items and then give two corollaries to show the performance of applying the indivisibly selling idea to distributing divisible items.

Consider the indivisible goods setting. For each agent-item pair $(i,j)$, define its weight $w_{ij}$ to be the maximum money that we can charge agent $i$ if assigning item $j$ to her, i.e., $w_{ij}=\min\{B_i,\frac{v_{ij}}{\tau_i}\}$. Since items are indivisible and each agent only wants to buy at most one item, a feasible solution is essentially a matching between the agent set and the item set, and the goal is to find a maximum weighted matching.
However, the algorithm to output the maximum weighted matching is not truthful. 
We observe that a natural greedy matching algorithm can return a constant approximation while retaining the truthfulness.
The mechanism is described in~\cref{alg:multi_unit_indivisible}.                

\begin{algorithm}[tb]
	\caption{\; Indivisibly Selling }
	\label{alg:multi_unit_indivisible} 	
	\begin{algorithmic}[1] %[1] enables line numbers
		\Require The reported budgets $\{B_i\}_{i\in [n]}$, the reported value profile $\{v_{ij}\}_{i\in[n],j\in[m]}$ and the reported target ratios $\{\tau_i\}_{i\in [n]}$.
		\Ensure An allocation and payments.
            \State Initially, set allocation $x_{ij}\leftarrow 0$ $\forall i\in [n],j\in [m]$ and payment $p_i\leftarrow 0$ $\forall i\in [n]$.
		\State Sort all the agent-item pairs $\{(i,j)\}_{i\in[n],j\in [m]}$ in the decreasing lexicographical order of $\left(\min\{B_i,\frac{v_{ij}}{\tau_i}\}, v_{ij}\right)$ and break the ties in a fixed manner.
		\For{each agent-item pair $(i,j)$ in the order}
		\State If both agent $i$ and item $j$ have not been matched, match them: $x_{ij}\leftarrow 1, p_i \leftarrow \min \{B_i,\frac{v_{ij}}{\tau_i}\}$.
		\EndFor
		\State \Return Allocation $\{x_{ij}\}_{i\in[n],j\in [m]}$ and payments $\{p_i\}_{i\in[n]}$.
	\end{algorithmic}
\end{algorithm}

\begin{theorem}\label{thm:multi_unit_indivisible}
	\cref{alg:multi_unit_indivisible} is feasible, truthful and achieves an approximation ratio of $1/2$ when items are indivisible.
\end{theorem}

\begin{proof}
	The feasibility is obvious since $\min \{B_i,\frac{v_{ij}}{\tau_i}\}$ is the maximum willingness-to-pay of agent $i$ when adding $(i,j)$ into the matching.      
    %For the truthfulness, a key observation is that 
	To prove the truthfulness, we show that once an agent misreports the profile and obtains a higher value, either the budget constraint or the RoS constraint must be violated. Since the agent-item pairs are sorted in the decreasing lexicographical order of $\left(\min\{B_i,\frac{v_{ij}}{\tau_i}\}, v_{ij}\right)$, the matched item value of agent $i$ is non-increasing when none of the related agent-item pairs are ranked higher. Thus, once the agent misreports a profile $(B_i',
 \vecv_i',\tau_i')$ and gets assigned an item $k$ with a higher value, the rank of pair $(i,k)$ must get improved, implying that $\min\{B_i',\frac{v'_{ik}}{\tau_i'}\} > \min\{B_i,\frac{v_{ik}}{\tau_k}\}$. Since the mechanism charges this agent $\min\{B_i',\frac{v'_{ik}}{\tau_i'}\}$ under the new reported profile, either the budget constraint or the RoS constraint must be unsatisfied.
 
 %such that the rank of pair $(i,j)$ gets improved. 
 %On the other hand, decreasing $\tau_i$ can improve the rank of pair $(i,j)$ only when $\frac{v_{ij}}{\tau_i} < B_i$ currently. Thus, once pair $(i,j)$ gets a higher rank with $\tau_i'<\tau_i$, we have 
%	\[ \min\{B_i,\frac{v_{ij}}{\tau_i'}\} > \frac{v_{ij}}{\tau_i}, \]
%	implying that the RoS constraint will be violated if the mechanism matches this pair. 
	
	Finally, we prove the approximation ratio by the standard analysis of the greedy matching algorithm. For each pair $(i,j)$ in an optimal matching, there must exist a pair (either $(i,j')$ or $(i',j)$) in the greedy matching whose weight is at least $c_{ij}$. Thus, the maximum matching weight is at most twice the weight of our matching, and \cref{alg:multi_unit_indivisible} gets a $1/2$-approximation.

\end{proof}

%\cref{thm:multi_unit_indivisible} helps build the indivisible selling idea and gives the intuition of developing truthful mechanisms for the divisible-item setting. 
Consider a feasible solution $\vz=\{z_{ij}\}_{i\in [n],j\in [m]}$ (not necessarily truthful) for multiple \textbf{divisible} items auction among unit-demand agents. We assume that each unit-demand agent $i$ has at most one variable $z_{ij}>0$,  Define $\cW_j(\vz) := \sum_{i:z_{ij}>0} p_i$ to be the total payments related to item $j$. We observe the following two corollaries.

\begin{corollary}\label{cor:multi_unit_indivisible_1}
	If solution $\vz$ is $\alpha$-approximation and for any item $j$, $\max_{i\in [n]} \min \{\frac{v_{ij}}{\tau_i}z_{ij},B_i\}\geq \beta \cdot \cW_j(\vz)$, then running \cref{alg:multi_unit_indivisible} directly obtains an approximation ratio of $\frac{\alpha \beta}{2}$. 
\end{corollary}

\begin{corollary}\label{cor:multi_unit_indivisible_2}
	For a constant $\beta\in [0,1]$, define item subset $H(\vz,\beta)\subseteq [m]$ to be the set of items with $\max_{i\in [n]} \min \{\frac{v_{ij}}{\tau_i}z_{ij},B_i\}\geq \beta \cdot \cW_j(\vz)$. Running \cref{alg:multi_unit_indivisible} directly obtains a total payment at least $\frac{\beta}{2}\sum_{j\in H(\vz,\beta)}\cW_j(\vz)$ for any $\beta \in [0,1]$.
\end{corollary}

\subsection{Foundations of Random Sampling}\label{subsec:random}

%This subsection introduces the random sampling method. As mentioned above, 
%As suggested in the random sampling procedure of~\cref{alg:single_all_pvt}, 
%The key idea of random sampling is gathering information from several randomly sampled agents and using it to guide the auction among all other agents. 
The subsection explores generalizing the random sampling procedure in~\cref{alg:single_all_pvt} to multiple items auction. 
We first randomly sample half of the agents and investigate how much revenue can be earned per unit of each item if the auctioneer only sells the items to these sampled agents. Recall that the mechanism does not actually distribute any item to the sampled agents. 
Then, the auctioneer sets the reserve price of each item based on the investigated revenues and sells them to all the remaining agents. 
More specifically, let these agents arrive in a fixed order. 
When an agent arrives, she is allowed to buy any remaining fraction of any item as long as she can afford the reserve price. 

It is easy to observe that the mechanism is still truthful according to the same argument in the proof of~\cref{thm:single_all_pvt}: for a sampled agent, she will not be assigned anything, and therefore, she does not have any incentive to lie; while for the agents that do not get sampled, they are also truthtelling because neither the arrival order nor the reserve prices are determined by their reported profiles and a fake profile that can improve the agent's obtained value must violate at least one constraint.

The key condition that random sampling can achieve a constant approximation ratio is that the revenue earned by each item among the sampled agents is (w.h.p.) close to its contribution to the objective in an optimal solution or a constant approximation solution; otherwise, there is no reason that the reserve prices are set based on the investigated revenues. Unfortunately, unlike the single-item environment, we cannot guarantee that an optimal solution of the multiple items auction satisfies this condition. 
Thus, to obtain such a nice structural property, we present an algorithm based on greedy matching and item supply clipping in~\cref{alg:greedy}. Note that this algorithm is untruthful, and we only use it to simulate the auction among the sampled agents. We first prove that it obtains a constant approximation, and then show several nice structural properties of the algorithm.

\begin{algorithm}[tb]
	\caption{\; Greedy Matching and Item Supply Clipping }
	\label{alg:greedy} 	
	\begin{algorithmic}[1] %[1] enables line numbers
		\Require The budgets $\{B_i\}_{i\in [n]}$, the value profile $\{v_{ij}\}_{i\in[n],j\in[m]}$ and the target ratios $\{\tau_i\}_{i\in [n]}$.
		\Ensure An allocation and payments.
            \State Initially, set allocation $x_{ij}\leftarrow 0$ $\forall i\in [n],j\in [m]$ and payment $p_i\leftarrow 0$ $\forall i\in [n]$.
		\State For each agent-item pair $(i,j) \in[n] \times [m]$, define its weight $w_{ij}:=\frac{v_{ij}}{\tau_i}$. Sort all the pairs in the decreasing order of their weights and break the ties in a fixed manner.
		\For{each agent-item pair $(i,j)$ in the order}
		\State If agent $i$ has not bought any item and the remaining fraction $R_j$ of item $j$ is more than $1/2$, $x_{ij}\leftarrow \min\{ R_j, \frac{B_i}{w_{ij}} \}, p_i \leftarrow  w_{ij}x_{ij}$.
		\EndFor
		\State \Return Allocation $\{x_{ij}\}_{i\in[n],j\in [m]}$ and payments $\{p_i\}_{i\in[n]}$.
	\end{algorithmic}
\end{algorithm}

\begin{theorem}\label{thm:greedy}
	The approximation ratio of \cref{alg:greedy} is $1/6$.
\end{theorem}
\begin{proof}
	
	Use $\left(\vx^*=\{x_{ij}^*\},\vp^*=\{p_i^*\} \right)$ and $\left(\vx=\{x_{ij}\}_{i\in [n], j\in [m]}, \vp=\{p_i\} \right)$ to represent the allocations and the payments in an optimal solution and \cref{alg:greedy}'s solution respectively. Without loss of generality, we can assume that $p_i^*=\sum_{j\in [m]}x_{ij}^*w_{ij}\leq B_i$ for any $i\in [n]$.    
	For each item $j\in [m]$, define $A_j$ to be the set of agents who buy some fractions of item $j$ in the optimal solution, i.e., $A_j := \{ i\in [n] \mid x_{ij}^*>0 \}$, and then based on $\vx$, we partition $A_j$ into three groups:
	\begin{equation*}
		\begin{aligned}
			&A_j^{(1)} = \{ i\in [n]  \mid  x_{ij} > 0 \}, \\
			&A_j^{(2)} = \{ i\in [n]  \mid x_{ij} = 0 \text{ due to } R_j\leq 1/2 \},\\
			&A_j^{(3)} = \{ i\in [n] \mid   x_{ij} = 0 \text{ due to agent $i$ has bought another item}\}.
		\end{aligned}
	\end{equation*}
	Note that if some agent does not buy the item $j$ in $\vx$ due to both of the two reasons,  we add the agent into an arbitrary one of $A_j^{(2)} $ and $A_j^{(3)} $.  
	
	Use $\opt$ and $\alg$ to denote the objective values of the optimal solution and our solution, respectively. Based on the partition mentioned above, we split the optimal objective into three parts:
	\[ \opt = \sum_{i\in [n], j\in [m]} x_{ij}^*w_{ij} =  \sum_{j\in [m]} \sum_{i\in A_j^{(1)}} x_{ij}^*w_{ij} +  \sum_{j\in [m]} \sum_{i\in A_j^{(2)}} x_{ij}^*w_{ij} +\sum_{j\in [m]} \sum_{i\in A_j^{(3)}} x_{ij}^*w_{ij} .\]
	In the following, we analyze the three parts one by one and show that each part is at most twice $\alg$, which implies that $\alg$ is $1/6$ approximation. 
	
	Due to the definition of $A_j^{(1)}$, for each $(i,j)$ pair in the first part, \cref{alg:greedy} assigns some fractions of item $j$ to agent $i$, and therefore, $x_{ij}\geq \min \{\frac{1}{2},\frac{B_i}{w_{ij}}\}$. Since $x_{ij}^*\leq 1$ and we assume w.l.o.g. that $x_{ij}^*\leq \frac{B_{i}}{w_{ij}}$, we have 
	\begin{equation}\label{eq:4-1}
		\begin{aligned}
			\sum_{j\in [m]} \sum_{i\in A_j^{(1)}} x_{ij}^*w_{ij} &\leq  \sum_{j\in [m]} \sum_{i\in A_j^{(1)}} \min \{w_{ij},B_i\} \leq \sum_{j\in [m]} \sum_{i\in A_j^{(1)}} 2w_{ij}\min\{\frac{1}{2},\frac{B_i}{w_{ij}}\} \\
			&\leq \sum_{j\in [m]} \sum_{i\in A_j^{(1)}} 2x_{ij}w_{ij}  \leq 2\alg .
		\end{aligned}
	\end{equation}
	
	For each item $j$ with non-empty $A_j^{(2)}$, \cref{alg:greedy} must sell at least half of the item, and then due to the greedy property of the algorithm, we have
	\[ \sum_{i\in A_j^{(2)} } x_{ij}^*w_{ij} \leq 2\cW_j(\vx), \]
	recalling that $\cW_j(\vx)=\sum_{i:x_{ij}>0}p_i$. Thus,
	\begin{equation}\label{eq:4-2}
		\begin{aligned}
			\sum_{j\in [m]} \sum_{i\in A_j^{(2)}} x_{ij}^*w_{ij} \leq  \sum_{j\in [m]} 2\cW_j(\vx)
			\leq 2\alg .
		\end{aligned}
	\end{equation}
	
	Finally, for each item $j$ and agent $i\in A_j^{(3)}$, suppose that agent $i$ buys some fractions of item $j'$ in solution $\vx$. Due to the greedy property, $w_{ij} \leq w_{ij'}$. Hence,
	\[ x^*_{ij}w_{ij} \leq \min \{B_i,w_{ij'}\} \leq 2 \min\{\frac{B_i}{w_{ij;}},\frac{1}{2}\} w_{ij'} \leq 2x_{ij'}w_{ij'}=2p_i. \]
	Summing over these $(i,j)$ pairs, 
	\begin{equation}\label{eq:4-3}
		\begin{aligned}
			\sum_{j\in [m]} \sum_{i\in A_j^{(3)}} x_{ij}^*w_{ij} = \sum_{i\in [n]} \sum_{j:i\in A_j^{(3)}}  x_{ij}^*w_{ij} \leq \sum_{i\in [n]} 2p_i = 2\alg.
		\end{aligned}
	\end{equation}
	
	Combining \cref{eq:4-1}, \cref{eq:4-2} and \cref{eq:4-3} completes the proof.
	
\end{proof}
\begin{remark}
	Note that the item supply clipping parameter $1/2$ in \cref{alg:greedy} can be replaced by any other constant in $(0,1)$. By setting this parameter to be $\frac{\sqrt{2}}{1+\sqrt{2}}$, the algorithm can get an approximation ratio of $ 3+2\sqrt{2}$.
\end{remark}

For an agent subset $S\subseteq [n]$, use $\left(\vx^{S}, \vp^{S}\right)$ to denote the allocation and the payments if using \cref{alg:greedy} to distribute all the items to agents in $S$. %In addition, define $ \cQ_j(\vx) $ to be the average selling price of item $j$ in solution $\vx$, i.e., $\cQ_j(\vx) = \frac{ \cW_j(\vx) }{\sum_{i\in [n]}x_{ij}}$ if $\sum_{i\in [n]}x_{ij} >0$ and $0$ otherwise. 
We claim the following lemma.

\begin{lemma}\label{thm:greedy_property}
	For any agent subset $S\subseteq [n]$, we have
	\begin{itemize}
		\item agent payment monotonicity: $p_i^{S} \geq \frac{1}{2} p_i$, $\forall i\in S$.
		%\item selling price monotonicity: $\cQ_j(\vx^{S}) \leq \cQ_j(\vx)$, $\forall j\in [m]$.
		\item selling revenue monotonicity: $ \cW_j(\vx^{S}) \leq 2\cW_j(\vx)$, $\forall j\in [m]$.
	\end{itemize}    
\end{lemma}

Use $R_j(i,k)$ and $R_j^S(i,k)$ to denote the remaining fractions of item $j$ at the end of pair $(i,k)$'s iteration when running~\cref{alg:greedy} for all the agents and for the agent subset $S$, respectively. Note that if $i\not \in S$, the corresponding iterations are viewed as empty iterations. We first show a key lemma that helps prove the two properties.

\begin{lemma}\label{lem:greedy_property}
    Consider an agent $i$ and let $k$ and $k'$ be the items that she buys in $\vx$ and $\vx^{S}$ respectively. We have $\forall j\in [m]$,
        \[ \max\left\{ R_j(i,k),\frac{1}{2} \right\} \leq \max\left\{ R_j^S(i,k'),\frac{1}{2} \right\}. \]
\end{lemma}
\begin{proof}
    We first show that for any pair $(i,k)$ and any item $j$, 
    \begin{equation}\label{eq:greedy}
        \max\left\{ R_j(i,k),\frac{1}{2} \right\} \leq \max\left\{ R_j^S(i,k),\frac{1}{2} \right\},
    \end{equation}
    
    Assume for contradiction that~\cref{eq:greedy} is violated for some agent-item pairs. Let $(i,k)$ be the first such pair in the order stated in~\cref{alg:greedy}. Notice that in this iteration, only the remaining fraction of item $k$ could change.  We distinguish three cases:
    (1) $x_{ik}^S=0$,
        (2) $x_{ik}^S>0$ and $x_{ik} > 0$,
        and (3) $x_{ik}^S>0$ and $x_{ik} = 0$.
    With some abuse of notation, we use $R_j^-(i,k)$ (resp. $R_j^{S-}(i,k)$) to denote the remaining fraction of item $j$ at the beginning of the iteration.
    
    For case (1), the remaining fraction $R_k^S$ remains unchanged. Thus,
    \[\max\left\{ R_k^-(i,k),\frac{1}{2} \right\}\geq \max\left\{ R_k(i,k),\frac{1}{2} \right\} > \max\left\{ R_k^S(i,k),\frac{1}{2} \right\} = \max\left\{ R_k^{S-}(i,k),\frac{1}{2} \right\}, \]
    contradicting the assumption that $(i,k)$ is the first such pair.

    For case (2), we have $x_{ik}^S=\min\{R_k^{S-}(i,k), \frac{B_i}{w_{ik}} \} $ and $x_{ik}=\min\{R_k^{-}(i,k), \frac{B_i}{w_{ik}} \} $ according to the algorithm. If $x_{ik} = R_k^{-}(i,k)$, then clearly, $R_k^{-}(i,k)$ becomes $0$ and~~\cref{eq:greedy} certainly holds; while if $x_{ik}=\frac{B_i}{w_{ik}}$, we have 
    $x_{ik}^S \leq \frac{B_i}{w_{ik}}= x_{ik},$
    and $ R_k^S(i,k) = R_k^{S-}(i,k) - x_{ik}^S\geq R_k^{-}(i,k)-x_{ik}=R_k(i,k), $
    contradicting the definition of pair $(i,k)$.

    For case (3), if $x_{ik} = 0 $ is due to $R_j^-(i,k)<\frac{1}{2}$, it is impossible that~\cref{eq:greedy} gets violated. Hence, the only reason that $x_{ik} = 0 $, in this case, is that agent $i$ has bought another item $k'$. This implies that in the iteration of pair $(i,k')$, we have $x_{ik'}^S=0$ and $x_{ik'} > 0$. Since agent $i$ had not bought any item that time, the only reason for $x_{ik'}^S=0$ is that $R_{k'}^{S-}(i,k')< \frac{1}{2}$. Due to the definition of $(i,k)$ and the fact that $(i,k')$ is in front of $(i,k)$ in the order, we have
    \[ R_{k'}^{-}(i,k') \leq \max\left\{ R_{k'}^-(i,k'),\frac{1}{2} \right\} \leq \max\left\{ R_{k'}^{S-}(i,k'),\frac{1}{2} \right\} = \frac{1}{2}, \]
    contradicting to $x_{ik'} > 0$.

    Thus, \cref{eq:greedy} holds for any agent-item pair. Then due to the same argument in the analysis of case (3) above, we see that $(i,k')$ must be in front of $(i,k)$ in the order, implying that $R_j^S(i,k)\leq R_j^S(i,k')$. Finally,
    \[ \max\left\{ R_j(i,k),\frac{1}{2} \right\} \leq \max\left\{ R_j^S(i,k),\frac{1}{2} \right\} \leq \max\left\{ R_j^S(i,k'),\frac{1}{2} \right\}. \]

\end{proof}

\begin{proof}[Proof of~\cref{thm:greedy_property}]
       We build on~\cref{lem:greedy_property} to prove the two properties one by one. 

       Consider an agent $i\in S$ and let $k$ and $k'$ be the items that she buys in $\vx$ and $\vx^{S}$ respectively (w.l.o.g., we can assume that each agent always buys something by adding some dummy items with value $0$.).
     Due to~\cref{lem:greedy_property} and the greedy property of \cref{alg:greedy}, we have $ w_{ik'} \geq w_{ik}$. Thus,
     \begin{equation*}
         \begin{aligned}
             p_i^S & = w_{ik'}\cdot x^S_{ik'} \geq w_{ik'} \cdot \min\left\{\frac{B_i}{w_{ik'}}, \frac{1}{2} \right\} \;\;\;\;\;\mbox{[$x^S_{ik'}>0 \Rightarrow R_{k'}^S > \frac{1}{2}$]}  \\
                    & \geq w_{ik} \cdot \frac{1}{2} \cdot  \min\left\{\frac{B_i}{w_{ik}}, 1 \right\} \geq w_{ik}\cdot \frac{1}{2} \cdot x_{ik} \;\;\;\qquad\mbox{[$ R_{k} \leq 1$]} \\
                    & \geq \frac{1}{2}p_i,
         \end{aligned}
     \end{equation*}
	which proves the agent payment monotonicity.
	
	% Again, due to~\cref{lem:greedy_property} and the greedy property, for each pair $(i,j)$ with $x_{ij}^S>0$ but $x_{ij}=0$, the only reason is that in solution $\vx$, the remaining fraction of item $j$ in that iteration is less than $1/2$; otherwise, agent $i$ must buy an item with a larger weight. Consider the last iteration $(i,j)$ such that the remaining fraction of item $j$ is more than $1/2$
 
 % This proves the selling price monotonicity directly.
    Now we prove the selling revenue monotonicity.
    Consider an arbitrary item $j$. Use $A_j^S$ and $A_j$ to denote the agents who buy some fractions of item $j$ in solution $\vx^S$ and $\vx$, respectively. Further, let $l^S$ and $l$ be the last buyer in $A_j^S$ and $A_j$, respectively. According to the assignment rule in the algorithm, for each agent $i\in A_j^S \cap A_j \setminus \{l\}$,  we have
    \begin{equation}\label{eq:sell_revenue}
        x_{ij}^S \leq \frac{B_i}{w_{ij}} = x_{ij};
    \end{equation}
    while for agent $l$, similar to the analysis in the last paragraph,
    \[ x_{lj}^S \leq \min\left\{\frac{B_i}{w_{ij}},1\right\} \leq 2\min\left\{\frac{B_i}{w_{ij}},\frac{1}{2}\right\} \leq 2x_{il}. \]
    Thus, if $A_j^S \subseteq A_j$, clearly, we have
    \[ \cW_j(\vx^S)=\sum_{i\in A_j^S}x_{ij}^Sw_{ij} \leq 2\sum_{i\in A_j^S}x_{ij}w_{ij} \leq 2\cW_j(\vx). \]
    
    It remains to show the case that $A_j^S \setminus A_j \neq \emptyset$. For an agent $i\in A_j^S \setminus A_j $, we have $x_{ij}^S>0$ but $x_{ij}=0$. Again, due to~\cref{lem:greedy_property}, we see the only reason is that in the process of computing solution $\vx$, the remaining fraction of item $j$ in that iteration is less than $1/2$; otherwise, agent $i$ must buy an item with a larger weight in solution $\vx^S$. Then due to the greedy property, we have $\forall i\in A_j^S \setminus A_j $, $w_{ij}\leq \min_{i'\in A_j} w_{i'j} = w_{lj}.$
    Thus, the property can be proved:
    \begin{equation*}
        \begin{aligned}
            \cW_j(\vx^S) & =\sum_{i\in A_j^S\cap A_j \setminus \{l\} }x_{ij}^Sw_{ij} + \sum_{i\in A_j^S\setminus A_j }x_{ij}^Sw_{ij} + x_{lj}^Sw_{lj} \\
            & \leq \sum_{i\in A_j^S\cap A_j \setminus \{l\} }x^S_{ij}w_{ij} + \left(\sum_{i\in \left(A_j^S\setminus A_j \right) \cup \{l\} }x^S_{ij}\right)\cdot w_{lj} \\
            & \leq \sum_{i\in A_j^S\cap A_j \setminus \{l\} }x^S_{ij}w_{ij} + \left(1-\sum_{i\in A_j^S\cap A_j \setminus \{l\}  }x^S_{ij}\right)\cdot w_{lj} \\
            & \leq \sum_{i\in A_j^S\cap A_j \setminus \{l\} }x_{ij}w_{ij} + \left(1-\sum_{i\in A_j^S\cap A_j \setminus \{l\}  }x_{ij}\right)\cdot w_{lj} \;\;\;\;\mbox{[\cref{eq:sell_revenue}]} \\
            & \leq 2\cW_j(\vx),
        \end{aligned}
    \end{equation*}
    where the last inequality used the fact that at least half of the item has been sold out in solution $\vx$.

    %If $A_j^S \subseteq A_j$, clearly, we have
    %\[ \cW_j(\vx^S) \leq 2\cW_j(\vx), \]
    %because 
 %    If for any pair $(i,j)$, we have $x_{ij}^S>0$ and $x_{ij}>0$, 
	% Again, due to~\cref{lem:greedy_property} and the greedy property, if there exists a pair $(i,j)$ with $x_{ij}^S>0$ but $x_{ij}=0$, the only reason is that in the process of computing solution $\vx$, the remaining fraction of item $j$ in that iteration is less than $1/2$; otherwise, agent $i$ must buy an item with a larger weight. Consider the last 
	% Finally, the selling revenue monotonicity can be proved easily by the selling price monotonicity. If the sold fraction of item $j$ in solution $\vx$ is already at least $1/2$, we have
	% \[ \cW_j(\vx^S) = \cQ_j(\vx^S) \cdot \sum_{i\in [n]}x^S_{ij} \leq 2  \cQ_j(\vx) \cdot \frac{1}{2} \leq 2\cW_j(\vx).   \]
	% On the other hand, if the sold fraction of item $j$ in solution $\vx$ is less than $1/2$, according to the argument in the last paragraph, for each pair $(i,j)$, we have $x_{ij}^S\leq x_{ij}$, which completes the proof of the third property.

\end{proof}

By~\cref{thm:greedy_property}, we have the following corollary.
\begin{corollary}\label{cor:greedy}
	Randomly dividing all the agents with equal probability into set $S$ and $R$, we have %Consider any agent subset $S \subseteq [n]$. We have
	\[   \E\left(\sum_{j\in [m]} \cW_j(\vx^S) \right)  = \E\left(\sum_{i\in S} p^S_i \right) \geq  \frac{1}{2} \E\left(\sum_{i\in S} p_i \right) = \frac{1}{4} \sum_{i\in [m]} p_i\geq  \frac{1}{4}\sum_{j\in [m]} \cW_j(\vx). \]
\end{corollary}

\subsection{Final Mechanism}\label{subsec:final}
\begin{algorithm}[htb]
	\caption{\; (\emph{Auxiliary}) Multiple Items Auction for Unit Demand Agents }
	\label{alg:aux_unit_demand} 	
	\begin{algorithmic}[1] %[1] enables line numbers
		\Require The reported budgets $\{B_i\}_{i\in [n]}$, the reported value profile $\{v_{ij}\}_{i\in[n],j\in[m]}$ and the reported target ratios $\{\tau_i\}_{i\in [n]}$.
		\Ensure An allocation and payments.
            \State Initially, set allocation $x_{ij}\leftarrow 0$ $\forall i\in [n],j\in [m]$ and payment $p_i\leftarrow 0$ $\forall i\in [n]$.
		\State With probability of $45/47$, run~\cref{alg:multi_unit_indivisible}. \Comment{{\color{gray}Indivisibly Selling Procedure}}
		\Begin{With probability of $2/47$,} \Comment{{\color{gray}Random Sampling Procedure}}
		\State Run~\cref{alg:greedy} and use $\left(\vz=\{z_{ij}\}, \vp(\vz)=\{p_i(\vz)\} \right)$ to denote the output.
		\State Randomly divide all the agents with equal probability into set $S$ and $R$. 
		\State Run~\cref{alg:greedy} on the sampled set $S$ and use $\left(\vz^S = \{z_{ij}^S\}, \vp(\vz^S)=\{p_i(\vz^S)\}\right)$ to denote the output.
		\State For each item $j$, set the reserve price $r_j \leftarrow \frac{1}{12} \cW_j(\vz^S)$.
		\State Let agents in $R$ come in an arbitrarily fixed order, and when each agent $i$ comes, find the unique item $k(i)$ that she buys in solution $\vz$ and use $R_{k(i)}$ to denote the remaining fraction of the item. Set $x_{i,k(i)}\leftarrow \min\{ R_{k(i)}, \frac{B_i}{r_{k(i)}} \}, p_i \leftarrow  r_{k(i)}x_{i,k(i)}$ if $r_{k(i)}\leq w_{i,k(i)}$.
		\End
		\State \Return Allocation $\{x_{ij}\}_{i\in[n],j\in [m]}$ and payments $\{p_i\}_{i\in[n]}$.
	\end{algorithmic}
\end{algorithm}
This subsection states the final mechanism, which is a random combination of the indivisibly selling idea and the random sampling idea. To streamline the analysis, we first introduce an \emph{auxiliary mechanism} which is constant-approximate but not truthful, and then show it can be altered to a truthful mechanism by losing only a constant factor on the approximation ratio.

\begin{theorem}\label{thm:unit_demand_aux}
	\cref{alg:aux_unit_demand} obtains a constant approximation ratio.
\end{theorem}

Recollect that $H(\vz,\beta):=\{ j\in [m] \mid \max_{i\in [n]} \min \{\frac{v_{ij}}{\tau_i}z_{ij},B_i\}\geq \beta \cdot \cW_j(\vz) \}$ defined in \cref{cor:multi_unit_indivisible_2}.
To prove \cref{thm:unit_demand_aux}, we partition all the items into two sets: $ H(\vz,\frac{1}{144})$ and $\bH(\vz,\frac{1}{144}) = [m] \setminus H(\vz,\frac{1}{144})$. \cref{cor:multi_unit_indivisible_2} directly implies that the first procedure (\cref{alg:multi_unit_indivisible}) guarantees our objective value is at least a constant fraction of $\sum_{j\in H(\vz,\frac{1}{144})} \cW_j(\vz)$. 

\begin{lemma}\label{lem:unit_demand_procedure_1}
	The revenue obtained by the first procedure in \cref{alg:aux_unit_demand} is at least $ \frac{1}{288} \sum_{j\in H(\vz,\frac{1}{144})} \cW_j(\vz)$.
\end{lemma}

For the second procedure, we show that $\sum_{j\in \bH(\vz,\frac{1}{144})} \cW_j(\vz)$ can be bounded by the total payment obtained by this procedure. More specifically, we prove the following technical lemma.

\begin{lemma}\label{lem:unit_demand_procedure_2}
	The expected revenue obtained by the second procedure in~\cref{alg:aux_unit_demand} is at least 
	\[\frac{1}{192} \sum_{j\in \bH(\vz,\frac{1}{144})} \cW_j(\vz) - \frac{7}{96} \sum_{j\in H(\vz,\frac{1}{144})} \cW_j(\vz).\]
\end{lemma}

%We first assume that \cref{lem:unit_demand_procedure_2} is correct and give the proof of \cref{thm:unit_demand_aux}.
\begin{proof}%[Proof of~\cref{lem:unit_demand_procedure_2}]
	
	Let $F$ and $D$ be the set of items that are sold out and the set of agents that use up their budgets in our solution, respectively. According to~\cref{alg:aux_unit_demand}, for a pair $(i,k(i))$, if $i\notin D$ and $k(i)\notin F$,
	\begin{equation}\label{eq:4.3-1}
		w_{i,k(i)} < r_{k(i)} = \frac{1}{12}\cW_j(\vz^S).
	\end{equation}
	We observe two lower bounds of the objective value of our solution: $\alg \geq \sum_{j\in F} \frac{1}{12}\cW_j(\vz^S)$,
	% \[\alg \geq \sum_{j\in F} \frac{1}{12}\cW_j(\vz^S), \]
	and
	\begin{equation*}
		\begin{aligned}
			\alg &\geq \sum_{i\in D} B_i  \geq \sum_{j\notin F} \sum_{i\in D} z_{ij}w_{ij}  = \sum_{j\notin F} \left( \sum_{i\in R} z_{ij}w_{ij}- \sum_{i\in R\setminus D} z_{ij}w_{ij}  \right) \\
			& \geq \sum_{j\notin F} \max\left\{0,\left( \sum_{i\in R} z_{ij}w_{ij}- \frac{1}{12}\cW_j(\vz^S) \right)\right\}, \\       
		\end{aligned}
	\end{equation*}
	where the last inequality used~\cref{eq:4.3-1}.
	
	For simplicity, use $\cW_j(\vz\cap S)$ to denote $\sum_{i\in S} z_{ij}w_{ij}$. Combing the two lower bounds, we have
	\[ 2\alg \geq \sum_{j\in F} \frac{1}{12}\cW_j(\vz^S) + \sum_{j\notin F} \max\left\{0,\left( \cW_j(\vz\cap R)- \frac{1}{12}\cW_j(\vz^S) \right)\right\}, \]
	and thus,
	\begin{equation}\label{eq:alg_lower1}
		2\E(\alg) \geq \sum_{j\in [m]} \E\left(\1_{j\in F}\cdot \frac{1}{12}\cW_j(\vz^S) + \1_{j\notin F} \cdot \left( \cW_j(\vz\cap R)- \frac{1}{12}\cW_j(\vz^S) \right) \right),
	\end{equation}
	where $\1_{(\cdot)}$ is an indicator function of the event $(\cdot)$. 
	
	According to the definition of $\bH(\vz,\frac{1}{144})$, Chebyshev's inequality and the concentration lemma \cite[Lemma 2]{DBLP:journals/mor/ChenGL14}, for any item $j\in \bH(\vz,\frac{1}{144})$, we have
	\begin{equation}\label{eq:concen}
		\begin{aligned}
			\Pr[\frac{1}{3}\cW_j(\vz) \leq \cW_j(\vz\cap S) \leq \frac{2}{3}\cW_j(\vz)] \geq \frac{15}{16},
		\end{aligned}
	\end{equation}
	which implies that with high probability, 
	\begin{equation}\label{eq:alg_lower2}
		\begin{aligned}
			\cW_j(\vz\cap R)- \frac{1}{12}\cW_j(\vz^S) & \geq \frac{1}{3}\cW_j(\vz) -\frac{1}{12}\cW_j(\vz^S) \geq \frac{1}{12} \cW_j(\vz^S),
		\end{aligned}
	\end{equation}
	where the last inequality used the selling revenue monotonicity.
	Use $\Pi_j$ to denote the event that the sampled subset $S$ satisfies $\frac{1}{3}\cW_j(\vz) \leq \cW_j(\vz\cap S) \leq \frac{2}{3}\cW_j(\vz)$. Combining \cref{eq:alg_lower1} and \cref{eq:alg_lower2}, 
	\begin{equation}\label{eq:alg_lower3}
		\begin{aligned}
			2\E(\alg) &\geq \sum_{ j\in \bH(\vz,\frac{1}{144}) } \Pr[\Pi_j]\cdot \E\left(\1_{j\in F}\cdot \frac{1}{12}\cW_j(\vz^S) + \1_{j\notin F} \cdot \left( \cW_j(\vz\cap R)- \frac{1}{12}\cW_j(\vz^S) \right) \bigg| \; \Pi_j \right) \\
			& \geq \frac{1}{12}\cdot \sum_{j\in \bH(\vz,\frac{1}{144})} \Pr[\Pi_j]\cdot \E\left( \cW_j(\vz^S) \; \bigg|\; \Pi_j\right).
		\end{aligned}
	\end{equation}
	
	%\[ \cW_j(\vz\cap R)- \frac{1}{12}\cQ_j(\vz^S) \geq \frac \]
	%For an item $j\in [m]$, define $\cW_j(\vz\cap S):=\sum_{i\in S} z_{ij}w_{ij} $ to be the total payment related to item $j$ among the agent subset $S$ in solution $\vz$. 

	We continue to find a lower bound of $\sum_{j\in \bH(\vz,\frac{1}{144})} \Pr[\Pi_j]\cdot \E\left( \cW_j(\vz^S) \; \bigg|\; \Pi_j\right)$. Observe that 
	\begin{equation*}
		\begin{aligned}
			\sum_{j\in [m]} \E\left( \cW_j(\vz^S) \right)  = &\sum_{j\in H(\vz,\frac{1}{144})} \E\left( \cW_j(\vz^S) \right) + \sum_{j\in \bH(\vz,\frac{1}{144})} \E\left( \cW_j(\vz^S) \right) \\
			=& \sum_{j\in H(\vz,\frac{1}{144})} \E\left( \cW_j(\vz^S) \right) + \sum_{j\in \bH(\vz,\frac{1}{144})} \Pr[\Pi_j]\cdot \E\left( \cW_j(\vz^S) \; \bigg|\; \Pi_j\right) \\
			&+ \sum_{j\in \bH(\vz,\frac{1}{144})} \Pr[\urcorner \Pi_j]\cdot \E\left( \cW_j(\vz^S) \; \bigg|\; \urcorner \Pi_j\right) \\
			\leq & \sum_{j\in H(\vz,\frac{1}{144})} \E\left( \cW_j(\vz^S) \right) + \sum_{j\in \bH(\vz,\frac{1}{144})} \Pr[\Pi_j]\cdot \E\left( \cW_j(\vz^S) \; \bigg|\; \Pi_j\right) \\
			&+ \frac{1}{16} \cdot \sum_{j\in \bH(\vz,\frac{1}{144})} \E\left( \cW_j(\vz^S) \; \bigg|\; \urcorner \Pi_j\right) \;\;\;\; \;\;\mbox{[\cref{eq:concen}]}\\
			\leq &  \sum_{j\in \bH(\vz,\frac{1}{144})} \Pr[\Pi_j]\cdot \E\left( \cW_j(\vz^S) \; \bigg|\; \Pi_j\right) + 2\sum_{j\in H(\vz,\frac{1}{144})}  \cW_j(\vz)  \\
			& + \frac{1}{8} \sum_{j\in \bH(\vz,\frac{1}{144})} \cW_j(\vz) \;\;\;\; \;\;\mbox{[selling revenue monotonicity]}\\
		\end{aligned}
	\end{equation*}

	Combining the above inequality and \cref{cor:greedy}, we get the lower bound:
	\begin{equation*}
		\begin{aligned}
			\sum_{j\in \bH(\vz,\frac{1}{144})} \Pr[\Pi_j]\cdot \E\left( \cW_j(\vz^S) \; \bigg|\; \Pi_j\right) &+ 2\sum_{j\in H(\vz,\frac{1}{144})}  \cW_j(\vz)   + \frac{1}{8} \sum_{j\in \bH(\vz,\frac{1}{144})} \cW_j(\vz) \geq \frac{1}{4} \sum_{j\in [m]} \cW_j(\vz)\\
			\sum_{j\in \bH(\vz,\frac{1}{144})} \Pr[\Pi_j]\cdot \E\left( \cW_j(\vz^S) \; \bigg|\; \Pi_j\right) &\geq \frac{1}{8} \sum_{j\in \bH(\vz,\frac{1}{144})} \cW_j(\vz) - \frac{7}{4}\sum_{j\in H(\vz,\frac{1}{144})}  \cW_j(\vz) .
		\end{aligned}
	\end{equation*}
	Thus, due to \cref{eq:alg_lower3}, we complete the proof:
	\[ \E(\alg) \geq \frac{1}{192}\sum_{j\in \bH(\vz,\frac{1}{144})} \cW_j(\vz)- 
	\frac{7}{96}\sum_{j\in H(\vz,\frac{1}{144})}  \cW_j(\vz)  \]
\end{proof}

\begin{proof}[Proof of \cref{thm:unit_demand_aux}]
	Combing \cref{lem:unit_demand_procedure_1}, \cref{lem:unit_demand_procedure_2} and the probabilities set in~\cref{alg:aux_unit_demand}, 
	\begin{equation*}
		\begin{aligned}
			\E(\alg) &\geq \frac{45}{47} \cdot \frac{1}{288} \sum_{j\in H(\vz,\frac{1}{144})} \cW_j(\vz) + \frac{2}{47} \cdot \left( \frac{1}{192} \sum_{j\in \bH(\vz,\frac{1}{144})} \cW_j(\vz) - \frac{7}{96} \sum_{j\in H(\vz,\frac{1}{144})} \cW_j(\vz)  \right) \\
			& = \frac{1}{4512} \sum_{j\in [m]} \cW_j(\vz) \geq \frac{1}{27072} \opt,
		\end{aligned}
	\end{equation*}
	where the last inequality used~\cref{thm:greedy}. 
\end{proof}

%Now we prove \cref{lem:unit_demand_procedure_2}.

\begin{algorithm}[htb]
	\caption{\; Multiple Items Auction for Unit Demand Agents }
	\label{alg:unit_demand} 	
	\begin{algorithmic}[1] %[1] enables line numbers
		\Require The reported budgets $\{B_i\}_{i\in [n]}$, the reported value profile $\{v_{ij}\}_{i\in[n],j\in[m]}$ and the reported target ratios $\{\tau_i\}_{i\in [n]}$.
		\Ensure An allocation and payments.
            \State Initially, set allocation $x_{ij}\leftarrow 0$ $\forall i\in [n],j\in [m]$ and payment $p_i\leftarrow 0$ $\forall i\in [n]$.
		\State With probability of $45/53$, run~\cref{alg:multi_unit_indivisible}. \Comment{{\color{gray}Indivisibly Selling Procedure}}
		\Begin{With probability of $8/53$,} \Comment{{\color{gray}Random Sampling Procedure}}
		\State Randomly divide all the agents with equal probability into set $S$ and $R$. 
		\State Run~\cref{alg:greedy} on the sampled agent set $S$ and use $\left(\vz^S = \{z_{ij}^S\}, \vp(\vz^S)=\{p_i(\vz^S)\}\right)$ to denote the output.
		\State For each item $j$, set the reserve price $r_j \leftarrow \frac{1}{12} \cW_j(\vz^S)$.
		\State Let agents in $R$ come in an arbitrarily fixed order. When each agent $i$ comes, use $ R_j $ to denote the remaining fraction of each item $j$, and let her pick the most profitable item $$k(i):=\argmax\limits_{j: r_j\leq \frac{v_{ij}}{\tau_i}} v_{ij} \cdot \min\{\frac{B_i}{r_j}, R_j\}.$$ 
		\State Set $x_{i,k(i)}\leftarrow \min\{\frac{B_i}{r_{k(i)}}, R_{k(i)}\}, p_i\leftarrow \min \{ w_{i,k(i)} x_{i,k(i)},B_i\}$ if item $k(i)$ exists. 
		\End
		\State \Return Allocation $\{x_{ij}\}_{i\in[n],j\in [m]}$ and payments $\{p_i\}_{i\in[n]}$.
	\end{algorithmic}
\end{algorithm}	

Finally, we present our final mechanism in~\cref{alg:unit_demand}. The only difference from~\cref{alg:aux_unit_demand} is that in the last step of the second procedure, we let the agent choose any item she wants as long as she can afford the reserve price, and then charge her the maximum willingness-to-pay.

\begin{theorem}\label{thm:unit_demand}
\cref{alg:unit_demand} is feasible, truthful, and constant-approximate.
\end{theorem}

\begin{proof}
According to \cref{thm:multi_unit_indivisible}, the first procedure is feasible and truthful. 
For the second procedure, the mechanism is truthful since any agent is charged her maximum willingness-to-pay. 
Then according to the same argument in the proof of~\cref{thm:single_all_pvt}, we can prove the truthfulness.

% as discussed at the beginning of~\cref{subsec:random}, neither agents in $S$ nor $R$ have the incentive to lie. Then due to the fact that any agent is charged her maximum willingness-to-pay, the second procedure is also feasible and truthful.

The following focuses on analyzing the approximation ratio.
To this end, we couple the randomness in ~\cref{alg:unit_demand} and~\cref{alg:aux_unit_demand}. 
The two algorithms are almost identical to each other except for one line and their randomness can be coupled 
perfectly. If by the coupling of randomness, both algorithms execute the first procedure, they are exactly identical and thus~\cref{lem:unit_demand_procedure_1}
also apples to ~\cref{alg:unit_demand}. 
Now, by randomness, they both execute the second procedure. In the second procedure, we can further couple the randomness so that they randomly sample the same set $S$. Conditional on all these (they both execute the second procedure and sample the same set $S$), we prove that  the revenue of  ~\cref{alg:unit_demand} is at least $\frac{1}{4}$ of that of ~\cref{alg:aux_unit_demand}. 

Let  $(\vx,\vp)$ and $(\vx',\vp')$ be the two solutions  respectively of 
~\cref{alg:unit_demand} and~\cref{alg:aux_unit_demand} under the above conditions. Let $ \alg$ and $\alg'$ be their revenues respectively.  
%We aim to show that if both~\cref{alg:unit_demand} and~\cref{alg:aux_unit_demand} execute the second procedure and obtain solution $(\vx,\vp)$ and $(\vx',\vp')$ respectively, we have $ \alg' \leq 4\alg$. Then due to~\cref{lem:unit_demand_procedure_1}, \cref{lem:unit_demand_procedure_2} and the probability parameters set in~\cref{alg:unit_demand}, a constant approximation ratio can be proved.
	Use $A_j'$ to denote the agents who buy some fractions of item $j$ in solution $\vx'$. According to~\cref{alg:aux_unit_demand}, \[\alg'  = \sum_{j\in [m]}\sum_{i\in A_j'}x_{ij}' r_j.\] 
 
        For an item $j$, if the corresponding revenue  in~\cref{alg:unit_demand}'s solution $\cW_j(\vx) \geq \frac{1}{2}r_j$, we have  
    $\sum_{i\in A_j'}x_{ij}' r_j \leq 2\cW_j(\vx)$,
    and then summing over all such items,
\begin{equation}\label{eq:unit_final_0}
    \sum_{j:\cW_j(\vx) \geq \frac{1}{2}r_j} \sum_{i\in A_j'}x_{ij}' r_j  \leq \sum_{j:\cW_j(\vx) \geq \frac{1}{2}r_j} 2\cW_j(\vx) \leq 2\alg.
\end{equation}
    
    For each item $j$ with $\cW_j(\vx) < \frac{1}{2}r_j$, we distinguish three cases for agents in $A_j'$ based on $(\vx,\vp)$: 
	(1) $p_i=B_i$,
		(2) $p_i<B_i$ and $x_{ij}>0$,
		and (3) $p_i<B_i$ and $x_{ij}=0$.

	For case (1), clearly, 
    \begin{equation}\label{eq:unit_final_1}
        x_{ij}'r_j \leq B_i = p_i .
    \end{equation}
    
 For case (2), since $\cW_j(\vx) < \frac{1}{2}r_j$, the remaining fraction of item $j$ is at least $1/2$ when~\cref{alg:unit_demand} let agent $i$ buy, and therefore, $x_{ij}\geq \min \{ \frac{1}{2},\frac{B_i}{ r_j } \}$.
 According to $p_i<B_i$, we have $p_i=w_{ij}x_{ij}\geq  r_j \min \left\{ \frac{1}{2},\frac{B_i}{ r_j } \right\}.$ Then, due to $x_{ij}'\leq \min  \{ 1,\frac{B_i}{ r_j } \}$, 
	\begin{equation}\label{eq:unit_final_2}
	    x_{ij}'r_j\leq 2p_i.
	\end{equation} 
	
	For case (3), suppose that agent $i$ buys item $k$ in solution $\vx$. Since  the remaining fraction of item $j$ is at least $1/2$ and the agent always pick the most profitable part in~\cref{alg:unit_demand}, we have
	\begin{equation*}
		\begin{aligned}
			\min\left\{\frac{1}{2}, \frac{B_i}{r_j}\right\} \cdot v_{ij}  \leq x_{ik}  v_{ik} \mbox{ and }
			\min\left\{\frac{1}{2}, \frac{B_i}{r_j}\right\} \cdot w_{ij}  \leq x_{ik}  w_{ik}.
		\end{aligned}
	\end{equation*}
	Again, due to $p_i<B_i$, $r_j \le w_{ij}$ and $x_{ij}'\leq \min  \{ 1,\frac{B_i}{ r_j } \}$, we have
	\begin{equation}\label{eq:unit_final_3}
	    \frac{1}{2}x_{ij}'r_j \leq  \min\left\{\frac{1}{2}, \frac{B_i}{r_j}\right\} \cdot w_{ij}  \leq x_{ik}  w_{ik} = p_i.
	\end{equation}
	
	Due to \cref{eq:unit_final_1}, \cref{eq:unit_final_2} and \cref{eq:unit_final_3}, for an item with $\cW_j(\vx) < \frac{1}{2}r_j$, in either case, we always have 
 $x_{ij}'r_j \leq 2p_i$. Thus, summing over all such items and the corresponding agents,
 \begin{equation}\label{eq:unit_final_4}
     \sum_{j:\cW_j(\vx) < \frac{1}{2}r_j} \sum_{i\in A_j'} x_{ij}'r_j \leq \sum_{j:\cW_j(\vx) < \frac{1}{2}r_j} \sum_{i\in A_j'} 2p_i \leq 2\alg.
 \end{equation}
 
 Combining \cref{eq:unit_final_0} and \cref{eq:unit_final_4} proves $\alg'\leq 4\alg$. 
Combining this with \cref{lem:unit_demand_procedure_2}, we know that 
	The expected revenue obtained by the second procedure in~\cref{alg:unit_demand} is at least 
	\[\frac{1}{768} \sum_{j\in \bH(\vz,\frac{1}{144})} \cW_j(\vz) - \frac{7}{384} \sum_{j\in H(\vz,\frac{1}{144})} \cW_j(\vz).\]
Further combining with ~\cref{lem:unit_demand_procedure_1}, which we argued also applies to~\cref{alg:unit_demand}, we have  
	the expected revenue obtained by~\cref{alg:unit_demand} is at least 
	\begin{equation*}
		\begin{aligned}
			& \frac{45}{53} \cdot \frac{1}{288} \sum_{j\in H(\vz,\frac{1}{144})} \cW_j(\vz) + \frac{8}{53} \cdot \left( \frac{1}{768} \sum_{j\in \bH(\vz,\frac{1}{144})} \cW_j(\vz) - \frac{7}{384} \sum_{j\in H(\vz,\frac{1}{144})} \cW_j(\vz)  \right) \\
			& = \frac{1}{5088} \sum_{j\in [m]} \cW_j(\vz) \geq \frac{1}{30528} \opt.
		\end{aligned}
	\end{equation*}
	
\end{proof}

\begin{remark} 
In the proof of \cref{thm:unit_demand}, we have not tried to optimize the constants in our analysis in the interests of expositional simplicity.
The parameters (e.g. $45/47$ and $1/144$) in our algorithm and analysis can be easily replaced by some other constants in $(0,1)$ to obtain another constant approximation ratio.
% There do exist some other probability distributions that achieve other constant approximation ratios.
% For example, we can replace the parameter $45/47$ and $1/144$ in our algorithm and analysis by some other constant in $(0,1)$ and obtain another constant approximation ratio.
\end{remark}

%neither the sampled agents nor the other agents have the incentive to misreport their target ratios. 

%According to the investigated revenue of each item, the auctioneer then set a reserve price of     

%random-sampling guanjian jiushi cong yige group li shouji xinxi, zhidao lingwaiyizu. suoyi yao yong greedy lai dedao yixie xinzhi.

%jialiangge corrallory

\section{Multiple Items Auction For Additive Agents}\label{sec:main_large}

This section studies the setting where the auctioneer has multiple items to sell and the agents are additive, that is, everyone can buy multiple items and obtain the sum value of the items. 
This environment is more challenging than the previous one, and some algorithmic ideas introduced in the last section are hard to apply. For example, one of the most critical components of~\cref{alg:unit_demand} is indivisibly selling, which is based on the observation that selling indivisible goods to unit-demand agents is much easier than selling divisible goods. However, this is not true in the additive valuation environment. To quickly see this, suppose that we have an approximation mechanism for selling indivisible goods to additive agents. Then we can obtain a mechanism with almost the same approximation ratio by splitting each item into tiny sub-items and selling them indivisibly. Thus, in the additive valuation environment, selling indivisible items is harder than selling divisible items.

Fortunately, we find that the idea of random sampling still works in this environment. Due to the relationship between our model and the liquid welfare maximizing model, the theoretical guarantee of the random sampling mechanism in~\cite{DBLP:conf/sagt/LuX17} directly implies a constant approximation of our problem under a large market assumption on the agents' budgets (that is, $B_i\leq \frac{\opt}{m \cdot c}$ for any agent $i$, where $c$ is a sufficiently large constant). This part is technically simple. We only state the theorem and the high-level idea here and defer some details to~\cref{sec:large}.

\begin{theorem}\label{thm:large_market}
    There exists a truthful constant approximation for multiple items auction among additive agents under the large market assumption.
\end{theorem}

\begin{proof}[Proof Sketch]
    For an instance $\cI = (\vB,\vecv,\vtau)$ of our model, we can easily construct a liquid welfare maximization instance $\cI'=(\vB',\vw')$, where for each agent $i$, the budget $B_i'=B_i$ and the valuation $w_{ij}'=\frac{v_{ij}}{\tau_i}$ $\forall j\in [m]$. 
    Since given the same allocation, the maximum willingness-to-pay of an agent in $\cI$ is exactly the agent's liquid welfare in $\cI'$, we see that the two instances share the same offline optimal objective values.
    
    Our mechanism is simply running the random sampling mechanism proposed in~\cite{DBLP:conf/sagt/LuX17}\footnote{See~\cref{sec:large} for the description of the mechanism} on the reduced instance $\cI'$.
    \cite{DBLP:conf/sagt/LuX17} showed that when the agents are quasi-linear utility maximizers subject to budget constraints, the total revenue obtained by the mechanism is at least a constant fraction of the optimal objective. %We build on this guarantee to prove the theorem.

    We note that the behavior of a value maximizer in the random sampling mechanism is different from the behavior of a quasi-linear utility maximizer. Thus, we cannot directly say that the proof has been completed due to $\opt(\cI)=\opt(\cI')$. The mechanism lets the agents come in a fixed order and allows each arrived agent to buy any fraction of the items she wants at the reserve prices. A quasi-linear utility maximizer will never buy any fraction of the items with reserve prices higher than the valuations (over the target ratio). However, a value maximizer may be interested in buying such items because the overall RoS constraint can still be satisfied even if for some items, the bought prices are higher than the valuations (over the target ratio). 

    We complete the proof by showing that the revenue obtained among value maximizers is always at least that obtained among quasi-linear utility maximizers. The key observation is that When an agent comes, regardless of the type, she computes a knapsack optimization problem with constraints. In other words, the agent sorts all the available items in the decreasing order of the ratio of valuation $w'_{ij}$ to the reserve price $r_j$, and then buys them sequentially as long as the constraints are satisfied. For a quasi-linear utility maximizer, she will keep buying until the budget is exhausted or all the remaining items have $\frac{w'_{ij}}{r_i}<1$; while a value maximizer may not stop immediately at the time that all the remaining items have $\frac{w'_{ij}}{r_i}<1$ when she still has budget left, instead, she will continue buying until the budget is used up or the overall RoS constraint is about to be violated. According to the above argument, it is easy to verify that the sold fraction of each item is non-decreasing when the agent becomes a value maximizer, and thus, the total revenue obtained among value maximizers is non-decreasing and can be bounded by a constant fraction of $\opt(\cI)$.

    %Note that although the same mechanism is used, due to the different agent types (the agents in $\cI'$ are quasi-linear utility maximizers subject to budget constraints), the agents' behaviors are different.   

    %Since given the same allocation, the maximum willingness-to-pay of an agent in $\cI$ is exactly the agent's liquid welfare in $\cI'$, we see that the two instances share the same offline optimal objective values, and therefore, the total revenue obtained among quasi-linear utility maximizers is guaranteed to be 

    %The mechanism is constructed by using a simple reduction from our instance $(\vB,\vecv,\vtau)$ to a liquid welf
    %Suppose that there is a random sampling mechanism $\cM$ for the liquid welfare maximizing model, whose input is the budget profile $\vB=\{B_i\}_{i\in [n]}$ and the value profile $\vw = \{w_{ij}\}_{i\in [n],j\in [m]}$. Our mechanism $\cM'$ is constructed as follows. Given an input profile $(\vB,\vecv,\vtau)$, define $\vw = \left\{w_{ij}=\frac{v_{ij}}{\tau_i}\right\}_{i\in [n], j\in [m]}$. Then, run mechanism $\cM$ on the input $(\vB,\vw)$ to get the allocation $\vx$. Finally, we charge each agent $i$ $p_i=\min\left\{B_i, \frac{ \sum_{j}v_{ij}x_{ij} }{\tau_i} \right\}$. 
\end{proof}

\section{Conclusion and Open Problems}\label{sec:con}

We investigate the emerging value maximizer in recent literature but also significantly depart from their modeling. 
We believe that the model and benchmark proposed in this paper are, on the one hand, more realistic and, on the other hand, friendlier to the AGT community. 
We get a few non-trivial positive results which indicate that this model and benchmark is indeed tractable. 
There are also many more open questions left.  
For additive valuation, it is open if we can get a constant approximation. 
It is interesting to design a mechanism with a better approximation for the setting of the single item and unit demand since our current ratio is fairly large. 
We also want to point out that no lower bound is obtained in this model, and thus any non-trivial lower bound is interesting. 
We get a much better approximation ratio for the single-item environment when valuation and budget are public than in the fully private setting. 
However, this is not a separation since we have no lower bound. 
Any separation result for different information models in terms of public and private is interesting.

\section*{Acknowledgment}
Chenyang Xu and Pinyan Lu were supported in part by Science and Technology Innovation 2030 –``The Next Generation of Artificial Intelligence'' Major Project No.2018AAA0100900. Additionally, Chenyang Xu received support from the Dean's Fund of Shanghai Key Laboratory of Trustworthy Computing, East China Normal University.
Ruilong Zhang was supported by NSF grant CCF-1844890.
\newpage
\bibliographystyle{splncs04}
\bibliography{ref}

\newpage

\appendix
% \begin{appendices}
\section{Partially Private Setting}\label{sec:part}

%\chenyang{Mark}
%Single item; then additive

This section studies a partially private setting proposed by~\cite{DBLP:conf/sigecom/BalseiroDMMZ22} where the budgets and values are all public. We first show that a better constant approximation for the single item auction can be obtained when the budgets become public in~\cref{subsec:single_add}. Then we build on the new single item auction to give an $\upOmega(\frac{1}{\sqrt{n}})$ approximation for multiple items auction among additive agents with public budgets and values in~\cref{subsec:add}.

\subsection{Single Item Auction with Public Budgets}\label{subsec:single_add}

%Let us warm up by considering the case where the auctioneer has only one item to sell. 
This subsection improves upon the previous approximation in~\cref{thm:single_all_pvt} when the agents' budgets become public.
The high-level idea is similar to the uniform price auction for liquid welfare maximization proposed in~\cite{DBLP:conf/icalp/DobzinskiL14}, which is allocating the item according to the maximum selling price such that if all agents buy the item at this price per unit, the item is guaranteed to be sold out.
Such a selling price is referred to be a \emph{market clearing price}.
However, new truthfulness challenges arise when applying the market clearing price idea to our auction environment. For example, there may exist such a case that the market clearing price remains unchanged when some agent changes the reported profile. Then in this case, the agent may misreport a lower target ratio or a larger value to obtain more goods without violating any constraint. %See more details in~\cref{sec:liquid}. 

To solve this issue, we use a simple scaling technique to partition the agents into two levels according to their reported profile and let the agents in the lower level buy the item at the market clearing price while the agents in the higher level have to pay a slightly higher price. The agent who determines the market clearing price always stays in the lower level, and she can obtain more goods only if she jumps into the higher level by increasing the reported $\frac{v_i}{\tau_i}$. However, in that case, the agent needs to pay a higher price that violates her RoS constraint. Thus, the agent has no incentive to misreport a lower ratio. The detailed mechanism is stated in~\cref{alg:single}. This subsection aims to show the following theorem.

\begin{algorithm}[tb]
	\caption{\;Single Item Auction with Public Budgets}
	\label{alg:single} 	
	\begin{algorithmic}[1] %[1] enables line numbers
        \Require The budgets $\{B_i\}_{i\in [n]}$, the reported value profile $\{v_i\}_{i\in[n]}$, the reported target ratios $\{\tau_i\}_{i\in [n]}$, and a parameter $\epsilon>0$.
        \Ensure An allocation and payments.
	\State For each agent $i\in[n]$, round $v_i/\tau_i$ down slightly such that it is an exponential multiple of $1+\epsilon$, i.e., $w_i:=(1+\epsilon)^{\lfloor \log_{(1+\epsilon)} \frac{v_i}{\tau_i}\rfloor}$.
	\State \label{line:reindex}Reindex the agents such that $w_1\geq w_2\geq ... \geq w_n$ and break the ties in a fixed manner.
        \State For any index $k\in [n]$, define $B[k]:=\sum_{i=1}^{k}B_i $.
        \State Find the maximum index $k$ such that $B[k]\leq w_k$.
        \If{$B[k] > w_{k+1}$}
        \State Round $B[k]$ up to $C[k] := (1+\epsilon)^{\lceil \log_{(1+\epsilon)} B[k]\rceil}$
        \For{each agent $i=1,...,n$}
        \If{$i\leq k$}
        \State $x_i\leftarrow\frac{B_i}{(1+\epsilon)B[k] }$, $p_i\leftarrow x_i \cdot C[k]$.
        \Else  %\Comment{$i>k$}
        \State $x_i\leftarrow 0$, $p_i\leftarrow 0$.  
        \EndIf
        \EndFor
        %\State For each agent $i\leq k$, let $x_i\leftarrow\frac{B_i}{(1+\epsilon)B[k] }$ and $p_i\leftarrow \frac{B_i}{1+\epsilon}$.
        %\State For each agent $i>k$, let $x_i\leftarrow 0$ and $p_i\leftarrow 0$.
        \Else \Comment{$B[k] \leq w_{k+1}$}
        \For{each agent $i=1,...,n$}
        \If{$i\leq k$} \Comment{Note that $w_i \geq w_{k+1}$ in this case}
        \State If $w_i>w_{k+1}$, $x_i\leftarrow \frac{B_i}{(1+\epsilon)w_{k+1}}$, $p_i\leftarrow x_i \cdot (1+\epsilon)w_{k+1}$. 
        \State If $w_i=w_{k+1}$, $x_i\leftarrow \frac{B_i}{(1+\epsilon)w_{k+1}}$, $p_i\leftarrow x_i \cdot w_{k+1}$.
        \Else
        \State If $i=k+1$, $x_{k+1}\leftarrow \frac{1}{1+\epsilon}-\frac{B[k]}{(1+\epsilon)w_{k+1}}$, $p_{k+1}\leftarrow x_i \cdot w_{k+1}$.
        \State If $i>k+1$, $x_i\leftarrow 0$, $p_i\leftarrow 0$.
        \EndIf
        \EndFor
        \EndIf
        \State \Return Allocation $\{x_i\}_{i\in[n]}$ and payments $\{p_i\}_{i\in[n]}$.
	\end{algorithmic}
\end{algorithm}

\begin{theorem}\label{thm:single}
For any parameter $\epsilon>0$, \cref{alg:single} is truthful and achieves an approximation ratio of $\frac{1}{(1+\epsilon)(2+\epsilon)}$, which tends to $\frac{1}{2}$ when $\epsilon$ approaches $0$.
\end{theorem}

We first show that the allocation satisfies the budget constraint and the reported RoS constraint of each agent, then discuss the truthfulness, and finally give the analysis of the approximation ratio.

\begin{lemma}\label{lem:single_feasible}
    Given any $\vtau$, for each agent $i$, we have $p_i\leq B_i$ and $ \tau_i p_i \leq x_i v_i$.
\end{lemma}
\begin{proof}
    We discuss case by case. If $B[k]>w_{k+1}$, for an agent $i\leq k$, we have
    \[ p_i = \frac{B_i \cdot C[k]}{(1+\epsilon)B[k]} < B_i, \]
    and
    \[ \frac{p_i}{x_i} = C[k] \leq w_{k} \leq w_i\leq  \frac{v_i}{\tau_i}. \]
    The first inequality in the second formula used the fact that $k$ is an index with $B[k] \leq w_{k}$ and $w_{k}$ is an exponential multiple of $1+\epsilon$. For all other agents, obviously, the two constraints are satisfied since their payments are $0$.

    Consider the case that $B[k]\leq w_{k+1}$. For an agent $i\leq k$, clearly, the budget constraint is satisfied.  
    If $w_{i} > w_{k+1}$, since each $w_i$ is an exponential multiple of $1+\epsilon$, we have $w_i \geq (1+\epsilon)w_{k+1}$, and
    \[ \frac{p_i}{x_i} = (1+\epsilon)w_{k+1}  \leq w_i\leq  \frac{v_i}{\tau_i}. \]
    Otherwise, we have $w_{i} = w_{k+1}$ and
    \[ \frac{p_i}{x_i} = w_{k+1}  = w_i\leq  \frac{v_i}{\tau_i}. \]
    For agent $k+1$, the budget constraint holds because for index $k+1$, $B[k+1]>w_{k+1}$ (otherwise, $k$ is not the maximum index with $B[k]\leq w_k$). More specifically,
    \[ p_{k+1} = \frac{w_{k+1}-B[k]}{1+\epsilon} = \frac{B_{k+1}+w_{k+1}-B[k+1]}{1+\epsilon} < B_{k+1}.\]
    The RoS constraint is also easy to show:
    \[\frac{p_i}{x_i} \leq w_{k+1} \leq \frac{v_i}{\tau_i}.\]
    Finally, for all other agents, the two constraints are satisfied since their payments are $0$.
    
\end{proof}

Then we prove the truthfulness. Notice that changing the reported profile may change the indices of the agents in step~\ref{line:reindex}. To avoid confusion, we use agent $a$ to represent a certain agent.
We first show that any agent $a$ will not misreport a lower $\frac{v_a}{\tau_a}$ because when $\frac{v_a}{\tau_a}$ becomes smaller, $x_a$ cannot increase (\cref{lem:single_truthful_1}); and then build on the RoS constraints to prove the other hand (\cref{lem:single_truthful_2}).
% and let $x_a(\vtau), p_a(\vtau)$ be the corresponding allocation and payment, respectively, when the input ratio profile is $\vtau$.
\begin{lemma}\label{lem:single_truthful_1}
   For any agent $a$, $x_a$ is non-increasing as $\frac{v_a}{\tau_a}$ decreases.  
\end{lemma}

\begin{proof}
    Given a reported profile $(\vecv,\vtau)$, refer to $\bp = \max\{B[k],w_{k+1}\}$ as the market clearing price. Decreasing $\frac{v_a}{\tau_a}$ unilaterally may change the value of $k$, the top-$k$ agents $S$, the index $\pi(a)$ of agent $a$, and the market clearing price $\bp$. Use $k'$, $S'$, $\pi'(a)$ and $\bp'$ to denote the three terms respectively after decreasing $\frac{v_a}{\tau_a}$ to $\frac{v_a'}{\tau_a'}$. 
    Clearly, if the current index $\pi(a)$ is already larger than $k$, $x_a$ is either $0$ or $\frac{1}{1+\epsilon}-\frac{B[k]}{(1+\epsilon)w_{k+1}}$, and will not increase as $\frac{v_a}{\tau_a}$ decreases. Thus, we only need to consider the case that $\pi(a)$ is at most $k$, i.e., $x_a= \frac{B_a}{(1+\epsilon)\bp}$. 

    Due to the observation that $\min_{i\in S\setminus\{a\}} w_i \geq \min_{i\in S} w_i \geq \sum_{i\in S} B_i > \sum_{i\in S\setminus \{a\}} B_i$, we have $k'\geq k-1$ and $S\setminus \{a\} \subseteq S'$ after decreasing $\frac{v_a}{\tau_a}$. If $k'=k-1$, w.l.o.g., we can assume that the new index $\pi'(a)$ is $k'+1$ and the new market clearing price $\bp'$ is $w'_a$; otherwise, agent $a$ obtains nothing. Let agent $b$ be the $(k+1)$-th player when the reported profile is $(\vecv,\vtau)$. Since $\pi'(a)=k'+1=k$, agent $a$ still ranks higher than agent $b$, i.e., $w'_a\geq w_b$.
    Then according to the definition of $k'$, we see that the market clearing price decreases:
    \[\bp = \sum_{i\in S} B_i > w'_a = \bp'. \]
    Thus, 
    \[x_a' = \frac{1}{1+\epsilon} - \frac{\sum_{i\in S\setminus \{a\}}B_i}{(1+\epsilon)\bp'} < \frac{1}{1+\epsilon} - \frac{\sum_{i\in S\setminus \{a\}}B_i}{(1+\epsilon)\bp} = \frac{\sum_{i\in S} B_i -  \sum_{i\in S\setminus \{a\}}B_i }{(1+\epsilon)\bp} = x_a.\]

    For the case that $k' \geq k$, we claim that either agent $a$ or agent $b$ is contained in $S'$. Suppose that $b\notin S'$. Since only agent $a$ changes the reported profile, it is easy to verify that $k'=k$ and $S'=S$, implying that $\bp'=\bp=\max\{\sum_{i\in S}B_i, w_{b}\}$ and $x_a'=x_a$. If $b \in S'$, due to the fact that $\sum_{i\in S\setminus \{a\}} B_i + B_a +B_b > w_b$ (the definition of $k$), agent $a$ can not belong to $S'$. Without loss of generality, assume that $\pi'(a)=k'+1$ and $\bp'=w'_a$; otherwise, $x'_a=0$. 
    We also see that the market clearing price is non-increasing: $\bp \geq w_b \geq \bp'$. Thus,
    \[ x'_a = \frac{1}{1+\epsilon}-\frac{\sum_{i\in S'}B_i}{(1+\epsilon)\bp'} \leq \frac{1}{1+\epsilon}-\frac{\sum_{i\in S\setminus \{a\}} B_i + B_b}{(1+\epsilon)\bp} = \frac{\bp-\sum_{i\in S\setminus \{a\}} B_i - B_b}{(1+\epsilon)\bp}.\]
    Regardless of whether $\bp$ takes the value $w_b$ or $\sum_{i\in S}B_i$, we always have $\bp - \sum_{i\in S\setminus \{a\}} B_i - B_b < B_a$, which implies that $x'_a<x_a$ and completes the proof. 
    
\end{proof}

\begin{lemma}\label{lem:single_truthful_2}
Consider any agent $a$ and any $\frac{v_a'}{\tau_a'}> \frac{v_a}{\tau_a}$. If $x'_a >x_a$, then $v_a x'_a < \tau_a p'_a$.
\end{lemma}

\begin{proof}
Use $\bp$ to denote the market clearing price when the reported profile is $(\vecv,\vtau)$. 
Clearly, if $w_a > \bp$, we have $x'_a = x_a$ for any $\frac{v_a'}{\tau_a'}> \frac{v_a}{\tau_a}$. In other words, $x'_a>x_a$ happens only when $w_a \leq \bp$.

We distinguish two cases. First, if $w_a < \bp$, the current price of the item (for agent $a$) is at least $(1+\epsilon)w_a$. Noticing that increasing $\frac{v_a}{\tau_a}$ cannot decrease the price, we have
\[ \frac{p_a}{x_a} \geq (1+\epsilon)w_a > \frac{v_a}{\tau_a}. \]

For the case that $w_a=\bp$. Since \cref{alg:single} breaks the ties in a fixed manner, $x_a$ increases only when agent $a$ jumps to the higher level, i.e., $w_a' > \bp$. Thus, according to the payment rule, we still have
\[ \frac{p_a}{x_a} \geq (1+\epsilon)w_a > \frac{v_a}{\tau_a}. \]
\end{proof}

Combining~\cref{lem:single_truthful_1} and \cref{lem:single_truthful_2} proves the truthfulness of the mechanism.
Finally, we analyze the approximation ratio of the mechanism.

\begin{lemma}\label{lem:single_ratio}
    \cref{alg:single} is $\frac{1}{(1+\epsilon)(2+\epsilon)}$-approximation.
\end{lemma}

\begin{proof}
    The proof is technically simple and similar to the analysis in~\cite[Theorem 4.4]{DBLP:conf/icalp/DobzinskiL14}. Use $\opt$ and $\alg$ to denote the optimal payment and our payment respectively. We first give an upper bound of $\opt$ and then establish the relationship between the upper bound and $\alg$. For the top-$k$ agents, due to the budget constraints, the optimal mechanism charges them at most $B[k]$; while for all the remaining agents, due to the RoS constraints, the optimal mechanism charges them at most $\max_{i>k} \frac{v_i}{\tau_i}\leq (1+\epsilon)w_{k+1}$. Namely,
    \[\opt \leq B[k] + (1+\epsilon)w_{k+1}.\]

    Then we analyze $\alg$. If $B[k] > w_{k+1}$, our total payment is
    \[\alg = \sum_{i\in[k]} p_i = \sum_{i\in [k]} \frac{B_i\cdot C[k]}{(1+\epsilon)B[k]} \geq \frac{B[k]}{(1+\epsilon)} > \frac{w_{k+1}}{(1+\epsilon)};\]
    while if $B[k]\leq w_{k+1}$, the total payment is
    \[\alg = \sum_{i\in[k]} p_i + p_{k+1} \geq \frac{B[k]}{1+\epsilon} + \frac{w_{k+1}}{1+\epsilon} - \frac{B[k]}{1+\epsilon} = \frac{w_{k+1}}{1+\epsilon} \geq \frac{B[k]}{1+\epsilon}. \]
    Thus, in either case, we have
    \begin{equation*}
        \begin{aligned}
            (1+\epsilon) \alg + (1+\epsilon)^2\alg &> \opt\\
            \alg & > \frac{\opt}{(1+\epsilon)(2+\epsilon)}.
        \end{aligned}
    \end{equation*}
    
\end{proof}

\subsection{Multiple Items Auction for Additive Agents}\label{subsec:add}

\begin{algorithm}[tb]
	\caption{\; Multiple Items Auction for Additive Agents in the Partially Private Setting}
	\label{alg:additive} 	
	\begin{algorithmic}[1] %[1] enables line numbers
		\Require The value profile $\{v_{ij}\}_{i\in[n],j\in[m]}$, the budgets $\{B_i\}_{i\in [n]}$, the reported target ratios $\{\tau_i\}_{i\in [n]}$ and a parameter $\epsilon>0$.
		\Ensure An allocation and payments.
		\State For each agent $i$, divide the budget into $m$ sub-budgets: $\left\{B_{ij}=B_i \cdot \frac{v_{ij}}{\sum_{j'\in [m]} v_{ij'}}\right\}_{j\in[m]} $. 
            \For{each item $j$}
            \State Run~\cref{alg:single} on the input $\left( \{v_{ij}\}_{i\in[n]}, \{B_{ij}\}_{i\in [n]}, \{\tau_i\}_{i\in [n]}, \epsilon \right)$ and obtain the solution $ \left( \vz_j=\{z_{ij}\}_{i\in [n]}, \vp(\vz_j)=\{p_{i}(\vz_j)\}_{i\in [n]} \right) $
            \State Pick the $(k+1)$-th agent $a$ in the permutation generated by the last line. 
            \State If $z_{aj} < \frac{1}{2}$, set $z_{aj} \leftarrow 0, p_i(\vz_j)\leftarrow 0$. \Comment{Item Supply Clipping}
            \EndFor
            \For{each agent $i$}
            \State For each item $j$, define $T_i(j) := \{j'\in [m] \mid z_{ij'} \geq z_{ij}\}$ and $U_i(j) := z_{ij} \cdot \sum_{j'\in T_i(j)} v_{ij'}$.
            \State Define $h(i):= \argmax_{j\in [m]} U_i(j)$. For each item $j$, set $x_{ij}\leftarrow z_{i,h(i)}$ if $j\in T_i(h(i))$ and 0 otherwise. Then, let $p_i\leftarrow \min\left\{B_i, \frac{U_i(h(i))}{\tau_i} \right\}$.
            \EndFor
		\State \Return Allocation $\{x_{ij}\}_{i\in[n],j\in [m]}$ and payments $\{p_i\}_{i\in[n]}$.
	\end{algorithmic}
\end{algorithm}

In this subsection, we build on the aforementioned single-item auction to give a truthful mechanism for multiple-items auction. 
The mechanism is described in~\cref{alg:additive}. 
One critical part of the mechanism is that it splits the budget of each agent and runs~\cref{alg:single} for each item to get solution $\left(\vz,\vp(\vz)\right)$. We observe that although each single item auction is truthful individually, outputting $\left(\vz,\vp(\vz)\right)$ directly gives an untruthful mechanism. An agent may misreport a lower target ratio to obtain more value because even if for some item $j$, the RoS constraint is violated (i.e., $\exists j\in [m], \frac{v_{ij}z_{ij}}{ p_i(\vz_j) } < \tau_i$), it is possible that the overall RoS constraint still holds when summing over all items because the return-on-spend ratio $\frac{v_{ij}z_{ij}}{ p_i(\vz_j) }$ of each bought item $j$ is different. 

A natural idea to handle this issue is raising the purchase prices of some items for an agent to guarantee that the agent's return-on-spend ratio of each bought item equals $\min \limits_{j:p_i(\vz_j)>0} \frac{v_{ij}z_{ij}}{ p_i(\vz_j) }$ so that once the agent violates the RoS constraint on some item, the overall RoS constraint must also be violated. 
Following this line, since the purchase prices are raised, to maintain the budget constraints, we need to reduce the number of items assigned to each agent. 
Thus, in~\cref{alg:additive}, we introduce $T_i(j)$ and let agent $i$ buy at most $z_{ij}$ fraction of any item $j'\in T_i(j)$. Finally, to maximize the total revenue, the mechanism charges each agent her maximum willingness-to-pay.
%Thus, wo men weishenme yao zheyang chuli
%Intuitively, the agent with $\frac{\sum_{j}v_{ij}z_{ij}}{\sum_{j} p_i(\vz_j) } \gg \tau_i$ may decrease $\tau_i$ to obtain more items without worrying about the RoS constraint.
%Intuitively，%the agent with $\frac{\sum_{j}v_{ij}z_{ij}}{\sum_{j} p_i(\vz_j) }$ far larger than $\tau_i$ % may decrease $\tau_i$ to obtain more items without worrying about the RoS constraint, because the agent has already saved a lot of money.
We state the main theorem in the following.

\begin{theorem}\label{thm:additive} 
\cref{alg:additive} is feasible, truthful, and obtains an approximation ratio of $\upOmega(\frac{1}{\sqrt{n}})$ when the budget profile and the value profile are public.
\end{theorem}

\subsubsection{Feasibility and Truthfulness}

We start by proving the feasibility and the truthfulness of the mechanism.

\begin{lemma}\label{lem:add_feasible}
For each item $j\in [m]$, \cref{alg:additive} satisfies the unit item supply constraint: $\sum_{i\in [n]} x_{ij} \leq 1$. For each agent $i\in [n]$, the mechanism satisfies the budget constraint and the RoS constraint: $p_i\leq B_i$ and $\tau_ip_i \leq \sum_{j\in[m]}x_{ij}v_{ij}$.
\end{lemma}

\begin{proof}

For each item $j$, since $ \vz_j$ is the assignments returned by running~\cref{alg:single} and applying an item supply clipping, we have $\sum_{i\in [n]}z_{ij} \leq 1$. According to the definition of $T_i(h(i))$, for any item $j\in T_i(h(i))$, $z_{ij} \geq z_{i,h(i)}$ and thus, $x_{ij} = z_{i,h(i)} \leq z_{ij}$, proving that the unit item supply constraints are satisfied.

For each agent $i$, the mechanism charges her $\min\left\{ B_i,\frac{U_i(h(i))}{\tau_i} \right\}$. According to the definition of $U_i(h(i))$, we see that this is exactly the total value of the obtained items. Hence, the mechanism satisfies the budget constraint and the RoS constraint. 

\end{proof}

Similar to the last subsection, we use two lemmas to prove the truthfulness. %\cref{lem:truth_1} and \cref{lem:truth_2} correspond to the two conditions in~\cref{lem:truthful} respectively.

\begin{lemma}\label{lem:truth_1}
    For any agent $i$, $\sum_{j\in [m]}v_{ij}x_{ij}$ is non-increasing as $\tau_i$ increases.
\end{lemma}

\begin{proof}
    For each agent-item pair $(i,j)$, according to~\cref{lem:single_truthful_1}, $z_{ij}$ is non-increasing as $\tau_i$ increases, which implies that $U_i(j)$ is also non-increasing. Since $h(i)$ is the item that obtains the maximum value of $U_i(j)$, $U_i(h(i))$ is non-increasing. As mentioned above, $U_i(h(i))$ is exactly the total obtained value. Thus, we have $\sum_{j\in [m]}v_{ij}x_{ij} = U_i(h(i))$ is non-increasing as $\tau_i$ increases.
\end{proof}

\begin{lemma}\label{lem:truth_2}
    Consider any agent $i$ and any $\tau'_i < \tau_i$. If $\sum_{j\in [m]}v_{ij}x'_{ij}>\sum_{j\in [m]}v_{ij}x_{ij}$, then $\sum_{j\in [m]}v_{ij}x'_{ij}< \tau_i p'_i$.
\end{lemma}
\begin{proof}
    Consider an agent $i$ and any $\tau'_i < \tau_i$, if $\sum_{j\in [m]}v_{ij}x'_{ij}>\sum_{j\in [m]}v_{ij}x_{ij}$, there must exist at least one item $l\in T'_i(h'(i))$ such that $z'_{il}> z_{il}\geq 0$; otherwise, the agent cannot obtain more valuable items. According to~\cref{lem:single_truthful_2}, we have \[\frac{p_i'(\vz_l')}{z'_{il}} > \frac{v_{il}}{\tau_i}.\]

    Consider the following payment rule: for each item $j$, we charge the agent 
    \[q'_{ij} = x'_{ij} \cdot \frac{p_i'(\vz_l')}{z'_{il}} \cdot \frac{v_{ij}}{v_{il}}.\]
    Clearly, this payment rule violates the RoS constraint for any item $j$:
    \[ \frac{q'_{ij}}{x'_{ij}} = \frac{p_i'(\vz_l')}{z'_{il}} \cdot \frac{v_{ij}}{v_{il}} > \frac{v_{il}}{\tau_i} \cdot \frac{v_{ij}}{v_{il}} = \frac{v_{ij}}{\tau_i}, \]
    and thus,
    \[ \sum_{j\in [m]} q'_{ij} > \frac{\sum_{j\in[m]} v_{ij}x'_{ij}}{\tau_i}. \]

    Finally, we show that $p'_i = \min \left\{ B_i,\frac{U'_i(h'(i))}{\tau'_i} \right\} \geq \sum_{j\in [m]} q'_{ij}$. According to~\cref{lem:single_feasible}, the single item auction mechanism satisfies $p_i'(\vz_l')\leq B_{il}$ and $p_i'(\vz_l') \leq \frac{v_{il}z_{il}'}{\tau_i'}$. Thus, for each item $j\in T'_i(h'(i))$, due to $x_{ij}' \leq z'_{il}$ and $ \frac{B_{il}}{v_{il}} = \frac{B_{ij}}{v_{ij}} $, we have
    \[ q'_{ij} = x'_{ij} \cdot \frac{p_i'(\vz_l')}{z'_{il}} \cdot \frac{v_{ij}}{v_{il}} \leq B_{ij},  \]
    and  
    \[ q'_{ij} = x'_{ij} \cdot \frac{p_i'(\vz_l')}{z'_{il}} \cdot \frac{v_{ij}}{v_{il}} \leq x_{ij}'\cdot \frac{v_{ij}}{\tau_i'}. \]
    Summing over all the items,
    \[ \sum_{j\in [m]}q_{ij}' \leq \min \left\{ B_i,\frac{\sum_{j\in [m]}x_{ij}'v_{ij}}{\tau_i'} \right\} = p_i'\;, \]
    completing the proof.
    
\end{proof}

\cref{lem:truth_1} prevents an agent from misreporting a target ratio higher than the actual ratio since the agent is a value maximizer, while \cref{lem:truth_2} guarantees that the agent cannot misreport a ratio lower than the actual ratio because otherwise, her RoS constraint will be violated. Thus, combing these two lemmas proves the truthfulness\footnote{We can also claim that \cref{lem:truth_1} and \cref{lem:truth_2} immediately prove the truthfulness according to \cite[Lemma 2.1]{DBLP:conf/sigecom/BalseiroDMMZ22}}. 

\subsubsection{Approximation Ratio}

This subsection analyzes the approximation ratio of~\cref{alg:additive}. As mentioned above, at the beginning stage of the mechanism, we split the budget of each agent based on the value profile. 
To streamline the analysis, we consider the setting where each agent $i$ can only use the sub-budget $B_{ij}$ to buy some fractions of each item $j$. Use $\bopt$ to denote the optimal objective of this sub-budget constrained setting. According to the approximation ratio of~\cref{alg:single} (\cref{lem:single_ratio}) and the item supply clipping bar $1/2$, we have
\begin{equation}\label{eq:5-1}
    \begin{aligned}
        \sum_{i\in [n],j\in[m]} z_{ij} \cdot p_i(\vz_j) \geq \frac{1}{2(1+\epsilon)(2+\epsilon)} \cdot \bopt
    \end{aligned}
\end{equation}
for any $\epsilon>0$.
This inequality splits our proof into two parts. We first show that $\bopt$ is at least $\frac{1}{2\sqrt{n}+3} \cdot \opt$, and then establish the relationship between our objective value and $\sum_{i\in [n],j\in[m]} z_{ij} p_i(\vz_j)$.

\begin{lemma}\label{lem:add_ratio1}
    $\bopt \geq  \frac{1}{2\sqrt{n}+3} \cdot \opt$
\end{lemma}

\begin{algorithm}[tb]
	\caption{\; Greedy Matching for the Sub-budget Constrained Setting}
	\label{alg:greedy_split} 	
	\begin{algorithmic}[1] %[1] enables line numbers
		\Require The value profile $\{v_{ij}\}_{i\in[n],j\in[m]}$, the budgets $\{B_{ij}\}_{i\in [n],j\in [m]}$ and the reported target ratios $\{\tau_i\}_{i\in [n]}$.
		\Ensure An allocation and payments.
		\State For each agent-item pair $(i,j) \in[n] \times [m]$, define its weight $w_{ij}:=\frac{v_{ij}}{\tau_i}$. Sort all the pairs in the decreasing order of their weights and break the ties in a fixed manner.
		\For{each agent-item pair $(i,j)$ in the order}
		\State Let $R_j$ be the remaining fraction of item $j$.
            \State Set $x_{ij}\leftarrow \min\{ R_j, \frac{B_{ij}}{w_{ij}} \}$.
		\EndFor
            \State For each agent $i$, $p_{i} \leftarrow  \sum_{j\in [m]}w_{ij}x_{ij}$.
		\State \Return Allocation $\{x_{ij}\}_{i\in[n],j\in [m]}$ and payments $\{p_i\}_{i\in[n]}$.
	\end{algorithmic}
\end{algorithm}

\begin{proof}
    Instead of comparing $\bopt$ and $\opt$ directly, we introduce a simple greedy algorithm for the sub-budget constrained setting in~\cref{alg:greedy_split} and show that the objective $\greedy$ obtained by the algorithm is at least $\frac{1}{2\sqrt{n}+3} \cdot \opt$. 

    Use $\left(\vx,\vp \right)$ and $\left(\vx^*,\vp^* \right)$ to represent the solution of~\cref{alg:greedy_split} and the optimal solution (of the original setting) respectively. We partition all the agents into two groups: $S:=\{i\in [n] \mid p_i \geq B_i/ \sqrt{n} \}$ and $R:=\{i\in [n] \mid p_i < B_i/ \sqrt{n} \}$, and get an upper bound of $\opt$:
    \begin{equation}\label{eq:5_final_1}
        \opt = \sum_{i\in[n],j\in[m]} x_{ij}^*w_{ij} \leq \sum_{i\in S} B_i + \sum_{j\in [m]}\sum_{i\in R} x_{ij}^*w_{ij} \leq \sqrt{n} \cdot \greedy + \sum_{j\in [m]}\sum_{i\in R} x_{ij}^*w_{ij}  .
    \end{equation}

    %According to the definition of $S$, we have
    %\[ \sum_{i\in S} B_i\leq \sqrt{n} \cdot \greedy. \]
    The remaining part is to prove that $\sum_{j\in [m]}\sum_{i\in R} x_{ij}^*w_{ij}$ can also be bounded by $O(\sqrt{n}) \cdot \greedy$. 
    For each item $j$, define $a(j) := \argmax_{i\in R} w_{ij}$ to be the agent $i\in R$ with the maximum $w_{ij}$. Clearly,
    \begin{equation}\label{eq:5_final_2}
        \sum_{j\in [m]}\sum_{i\in R} x_{ij}^*w_{ij} \leq \sum_{j\in [m]} w_{a(j),j} \;.
    \end{equation}
    We further partition all the items into two groups based on their assignments in the greedy solution: $P:= \{j\in [m] \mid x_{a(j),j}w_{a(j),j} < B_{a(j),j} \} $ and $Q:=\{j\in [m] \mid x_{a(j),j}w_{a(j),j} = B_{a(j),j} \}$.

    For each item $j\in P$, if sorting all agents in the decreasing order of $\{w_{ij}\}$, agent $a(j)$ is either the last agent who buys item $j$ in~\cref{alg:greedy_split}, or ranks behind the last agent buying item $j$; otherwise, agent $a(j)$ must exhaust the sub-budget $B_{a(j),j}$. Thus, $w_{a(j),j} \leq \cW_j(\vx)$, and therefore,
    \begin{equation}\label{eq:5_final_3}
        \sum_{j\in P} w_{a(j),j} \leq \sum_{j\in P} \cW_j(\vx) \leq \greedy . 
    \end{equation}

    For the items in $Q$, we reorganize the corresponding formula:
    \begin{equation}\label{eq:5-2-0}
        \sum_{j\in Q} w_{a(j),j} = \sum_{i\in R} \; \sum_{j\in Q : a(j)=i} w_{ij}.
    \end{equation}
    For simplicity, use $Q(i)$ to denote the item subset $ \{j\in Q \mid a(j)=i  \} $.
    We aim to show that $\forall i\in R$, $\sum_{j\in Q(i)} w_{ij} $ is at most $\greedy / (\sqrt{n}-1)$, and thus, their sum can be bounded by $O(\sqrt{n}) \cdot \greedy$.
    For each agent $i\in R$, due to the similar argument in the last paragraph, we have
    \begin{equation}\label{eq:5-2}
        \sum_{j\notin Q(i)} w_{ij} \leq \sum_{j\notin Q(i)} \cW_j(\vx) \leq \greedy .         
    \end{equation}

    Recall that any agent $i \in R$ pays less than $\frac{B_i}{\sqrt{n}}$. It is easy to observe that for an agent $i\in R$, the sum budget of the items in $Q(i)$ is very limited because the agent spends very little compared to the budget even though she has exhausted the sub-budgets of these items. More formally, we have
    \begin{equation}\label{eq:5-3}
        \begin{aligned}
            \sum_{j\in Q(i)} B_{ij} & < \frac{B_i}{\sqrt{n}} \\
            \sum_{j\in Q(i)} B_i \cdot \frac{v_{ij}}{\sum_{j'\in [m]}v_{ij'}} & \leq \frac{B_i}{\sqrt{n}}  \\
            \frac{\sum_{j\in Q(i)} w_{ij}}{\sum_{j\in Q(i)} w_{ij} + \sum_{j\notin Q(i)} w_{ij}} & \leq \frac{1}{\sqrt{n}} \\
            \sum_{j\in Q(i)} w_{ij} & \leq \frac{1}{\sqrt{n}-1} \sum_{j\notin Q(i)} w_{ij} .
        \end{aligned}
    \end{equation}

    Combing~\cref{eq:5-2-0}, \cref{eq:5-2} and \cref{eq:5-3} and then summing over all agents in $R$, we have
    \begin{equation}\label{eq:5_final_4}
        \sum_{j\in Q} w_{a(j),j} = \sum_{i\in R} \; \sum_{j\in Q(i)} w_{ij} \leq \frac{n}{\sqrt{n}-1} \cdot \greedy . 
    \end{equation}

    Finally, combing~\cref{eq:5_final_1}, \cref{eq:5_final_2}, \cref{eq:5_final_3} and \cref{eq:5_final_4} completes the proof:
    \[ \opt \leq \left( \sqrt{n} + 1 + \frac{n}{\sqrt{n}-1} \right) \cdot \greedy \leq (2\sqrt{n} + 3) \bopt. \]
    %For all items in $Q(i)$, agent $i$ exhausts their sub-budgets in the greedy solution. However, the agent belongs to set $R$

    %define $A_j:\{i\in [n] \mid x_{ij}>0 \}$ to be the set of agents who buy some fractions of item $j$ in~\cref{alg:greedy_split}.

    %For each item $j\in P$, the last 

\end{proof}

\begin{lemma}\label{lem:add_ratio2}
    For any $\epsilon>0$, $\alg \geq \min\left\{\frac{1}{2},\frac{1}{1+\epsilon}\right\} \cdot \sum_{i\in [n],j\in[m]} z_{ij} p_i(\vz_j) $
\end{lemma}
\begin{proof}
    We prove the lemma by showing that for any agent $i$, $p_i\geq \min\left\{\frac{1}{2},\frac{1}{1+\epsilon}\right\}\cdot \sum_{j\in[m]} z_{ij} p_i(\vz_j) .$ 
    Consider an arbitrary agent $i$. Use $g(i)$ to denote the item $j$ with the minimum non-zero $z_{ij}$, i.e., $g(i):=\argmin_{j: z_{ij}>0} z_{ij}$. We construct an auxiliary allocation $\{y_{ij}\}_{j\in [m]}$ and payment $q_i$ as follows:
    \begin{itemize}
        \item For each item $j$, set $y_{ij}=z_{i,g(i)}$ if $j\in T_i(g(i))$ and 0 otherwise.
        \item Find the most cost-effective available item $l:=\argmax_{j\in T_i(g(i))} \frac{ p_i(\vz_j) }{z_{ij}v_{ij}}$ and set \[q_i = \sum_{j\in [m]} y_{ij}\cdot \frac{ p_i(\vz_l) }{z_{il}v_{il}} \cdot v_{ij} \;.\]
    \end{itemize}

    Similar with the last part analysis in the proof of~\cref{lem:truth_2}, we see that payment $q_i$ is at most $\min\{B_i,\frac{U_i(g(i))}{\tau_i}\}$, and therefore, 
    \[q_i \leq \min\left\{B_i,\frac{U_i(g(i))}{\tau_i}\right\}\leq \min \left\{B_i, \frac{U_i(h(i))}{\tau_i}\right\} = p_i,\]
    where the second inequality used the fact that $h(i):= \argmax_{j\in [m]} U_i(j)$.

    Now we show that $q_i$ is at least a constant fraction of $\sum_{j\in[m]} z_{ij} p_i(\vz_j)$. 
    Noting that $g(i)$ is the item with the minimum non-zero $z$-value, 
    \[\sum_{j\in[m]} z_{ij} \cdot p_i(\vz_j) = \sum_{j\in T_i(g(i)) } z_{ij} \cdot p_i(\vz_j).\]
    We distinguish two cases based on the value of $z_{i,g(i)}$: (1) $z_{i,g(i)} \geq 1/2$, (2) $z_{i,g(i)} < 1/2$.

    If $z_{i,g(i)} \geq 1/2$, we have
    \[ y_{ij}\cdot \frac{ p_i(\vz_l) }{z_{il}v_{il}} \cdot v_{ij} \geq  \frac{1}{2} \cdot \frac{ p_i(\vz_j) }{z_{ij}v_{ij}} \cdot v_{ij} \geq \frac{1}{2} \cdot z_{ij}\cdot p_i(\vz_j)  \]
    for any item $j\in T_i(g(i)) $.

    For the second case, due to the item supply clipping in~\cref{alg:additive}, agent $i$ must be one of the top-$k$ agents in the single-item auction that sells item $g(i)$. Thus, according to~\cref{alg:single}, we have $p_i(\vz_{g(i)}) \geq \frac{B_{i,g(i)}}{1+\epsilon}$. Thus, for any item $j\in T_i(g(i))$, 
    \begin{equation*}
        \begin{aligned}
            y_{ij}\cdot \frac{ p_i(\vz_l) }{z_{il}v_{il}} \cdot v_{ij} & \geq z_{i,g(i)}\cdot \frac{ p_i(\vz_{g(i)}) }{z_{i,g(i)}v_{i,g(i)}} \cdot v_{ij} \\
            & \geq \frac{B_{i,g(i)}}{1+\epsilon}\cdot \frac{v_{ij}}{v_{i,g(i)}} \\
            & = \frac{B_{ij}}{1+\epsilon} \\
            & \geq \frac{1}{1+\epsilon} \cdot z_{ij} \cdot p_i(\vz_j). 
        \end{aligned}
    \end{equation*}

    Thus, in either case, we have
    \[ p_i \geq q_i =\sum_{j\in [m]} y_{ij}\cdot \frac{ p_i(\vz_l) }{z_{il}v_{il}} \cdot v_{ij}\geq \min\left\{\frac{1}{2},\frac{1}{1+\epsilon}\right\} \cdot \sum_{j\in [m]} z_{ij} \cdot p_i(\vz_j),  \]
    which completes the proof.

\end{proof}

Combining~\cref{eq:5-1}, \cref{lem:add_ratio1} and \cref{lem:add_ratio2} proves an approximation ratio of $\upOmega(\frac{1}{\sqrt{n}} )$.
%\section{Comparison with Liquid Welfare Maximization Model}\label{sec:liquid}

\section{Omitted Details in Section~\ref{sec:main_large}}\label{sec:large}

In this section, we restate the random sampling mechanism proposed in \cite{DBLP:conf/sagt/LuX17} and the results they obtained. The mechanism is described in~\cref{alg:rand}.

\begin{theorem}[\cite{DBLP:conf/sagt/LuX17}]
The random sampling mechanism is a universal truthful budget feasible mechanism which guarantees a constant fraction of the liquid welfare under the large market assumption.
\end{theorem}

The correctness of the above theorem heavily depends on~\cite[Lemma 6]{DBLP:conf/sagt/LuX17}, which states that the liquid welfare obtained from the random sampling algorithm is at least some constant fraction of the optimal mechanism.
To prove~\cite[Lemma 6]{DBLP:conf/sagt/LuX17}, they use the revenue obtained by a truthful auction as a lower bound of the liquid welfare.
Thus, \cite[Lemma 6]{DBLP:conf/sagt/LuX17} actually holds for the revenue maximization objective.
Hence, we have the following corollary.

\begin{corollary}
The random sampling mechanism is a budget feasible and truthful mechanism which achieves a constant approximation under the large market assumption.
\end{corollary}

\begin{algorithm}[tb]
	\caption{\cite{DBLP:conf/sagt/LuX17}~Greedy Algorithm}
	\label{alg:gre} 	
	\begin{algorithmic}[1] %[1] enables line numbers
        \Require $n$ agents with valuations $\set{w_{ij}}_{i\in[n],j\in[m]}$ and corresponding budgets $\set{B_i}_{i\in[n]}$. 
        \Ensure An allocation $(x_{ij})_{i\in[n],j\in[m]}$.
	    \For{each $i\in[n]$}
	        \State $C_i \leftarrow B_i$;
	    \EndFor
        \For{each $j\in[m]$}
            \State $s_j \leftarrow 1$;
        \EndFor
        \For{Each $i\in [n]$ and $j\in[m]$}
            \State $x_{ij} \leftarrow 0$;
        \EndFor
        \For{each $w_{ij}>0$ in decreasing order}
            \If{$C_i>w_{ij}s_j$}
                \State $x_{ij} \leftarrow s_j$;
                \State $C_i \leftarrow C_i - w_{ij}s_j$;
                \State $s_j \leftarrow 0$;
            \Else
                \State $x_{ij} \leftarrow \frac{C_i}{w_{ij}}$;
                \State $s_j \leftarrow s_j - \frac{C_i}{w_{ij}}$;
                \State $C_i \leftarrow 0$.
            \EndIf
        \EndFor
        \State \Return Allocation $\{x_{ij}\}_{i\in[n],j\in[m]}$.
	\end{algorithmic}
\end{algorithm}

\begin{algorithm}[tb]
\caption{\cite{DBLP:conf/sagt/LuX17}~Random Sampling Mechanism}
\label{alg:rand}
\begin{algorithmic}[1]
\Require $n$ agents with valuations $\set{w_{ij}}_{i\in[n],j\in[m]}$ and corresponding budgets $\set{B_i}_{i\in[n]}$.
\Ensure All allocation $\set{x_{ij}}_{i\in[n],j\in[m]}$ and payments $\set{p_i}_{i\in[n]}$.
\State Randomly divide all agents with equal probability into groups $T$ and $R$.
\State $A^T \leftarrow$ the allocation obtained by running~\cref{alg:gre} on group $T$.
\For{$j\in[m]$}
\State Use $w(A_j^T)$ to denote the total revenue contributed by item $j$.
\State $p_j \leftarrow \frac{1}{6}w(A_j^T)$;
\EndFor
\State Each agent $i\in R$ comes in a given fixed order and buy the most profitable part with respect to price vector $\set{p_j}_{j\in[m]}$ under budget feasibility and unit item supply constraints.
\end{algorithmic}
\end{algorithm}

\end{document}